\documentclass[12pt,letterpaper]{article}
\usepackage{mathrsfs}
\usepackage{geometry}

\usepackage{graphicx}
\usepackage{amssymb}
\usepackage{epstopdf}
\usepackage{amsmath,amscd}
\usepackage{amsthm}
\theoremstyle{plain} \textwidth=430pt \textheight=660pt
\makeatletter

\newcommand{\Rmnum}[1]{\expandafter\@slowroman #1@}
\makeatother

\newtheorem{theorem}{Theorem}[section]
\newtheorem{definition}[theorem]{Definition}
\newtheorem{lemma}[theorem]{Lemma}
\newtheorem{example}[theorem]{Example}
\newtheorem{assumption}[theorem]{Assumption}
\newtheorem{proposition}[theorem]{Proposition}

\title{Local Statistics of Realizable Vertex Models}
\author{Zhongyang Li\footnote{Department of mathematics, Brown University,
Providence, RI 02912, USA, zli@math.brown.edu}}
\date{}

\begin{document}
\maketitle
\begin{abstract}
We study planar ``vertex'' models, which are probability measures on
edge subsets of a planar graph, satisfying certain constraints at
each vertex, examples including dimer model, and 1-2 model, which we
will define. We express the local statistics of a large class of
vertex models on a finite hexagonal lattice as a linear combination
of the local statistics of dimers on the corresponding Fisher graph,
with the help of a generalized holographic algorithm. Using an
$n\times n$ torus to approximate the periodic infinite graph, we
give an explicit integral formula for the free energy and local
statistics for configurations of the vertex model on an infinite
bi-periodic graph. As an example, we simulate the 1-2 model by the
technique of Glauber dynamics.
\end{abstract}
\section{Introduction}

A \textbf{vertex model} is a graph $G=(V,E)$ where we associate to
each vertex $v\in V$ a \textbf{signature} $r_v$. A \textbf{local
configuration} at a vertex $v$ is a subset of incident edges of $v$.
A \textbf{configuration} of the graph $G$ is an edge subset of $G$.
The signature $r_v$ at a vertex $v$ is a function which assigns a
nonnegative real number (\textbf{weight}) to each local
configuration at $v$. The \textbf{partition function} of the vertex
model is the weighted sum of configurations $X\in\{0,1\}^{E}$, where
the weight of a configuration is the product of weights of local
configurations, obtained by restricting the configuration at each
vertex. Dimers, loop models, and random tiling models are some
special examples of vertex models.

Direct computations of the partition function of a general vertex
model usually require exponential time. On the other hand, using the
Fisher-Kasteleyn-Temperley method \cite{ka1,ka2}, we can efficiently
count the number of perfect matchings (dimer configurations) of a
finite planar graph. The idea of \textbf{generalized holographic
reduction} is to reduce a vertex model on a planar graph to a dimer
model on another planar graph, essentially by a linear base change,
see section 3. For an original version of the holographic reduction
(Valiant's Algorithm), see \cite{val}.

However, not all satisfying assignment problems can be reduced to a
perfect matching problem (``realized'') using the holographic
algorithm. We study the realizability problem of the generalized
algorithm for vertex models on the hexagonal lattice and prove that
the signature of realizable models form a submanifold of positive
codimension of the manifold of all signatures, see Theorem 3.6. An
example of realizable models is the \textbf{1-2 model}, which is a
signature on the honeycomb lattice,  only one or two edges allowed
to be present in each local configuration, see Figure 13, 14. The
realizability problem of Valiant's algorithm is studied by Cai
\cite{cai}, and the realizability problem of \textbf{uniform 1-2
model} (not-all-equal relation), a special 1-2 model which assigns
all the configurations weight 1, under Valiant's algorithm is
studied by Schwartz and Bruck \cite{sb}.

Realizable vertex models may be reduced to dimer models in more than
one way, that is, using different bases. However, all the dimer
models corresponding to the same vertex model are shown to be
\textbf{gauge equivalent}, i.e. obtained from one another by a
trivial reweighting.

One of the simplest vertex configuration models is a graph with the
same signature at all vertices. Using the singular value
decomposition, we prove that such models on a hexagonal lattice are
realizable if and only if they are realizable under orthogonal base
change. Moreover, the orthogonal realizability condition takes a
very nice form; see section 3.2.

We compute the local statistics of realizable vertex models on a
hexagonal lattice with the help of the generalized holographic
reduction, i.e. for the natural probability measure, we compute the
probabilities of given configurations at finitely many fixed
vertices, which are proved to be computable by sums of finitely many
Pfaffians, see Theorem 5.1 and Theorem 5.2.

The weak limit of probability measures of the vertex model on finite
graphs are of considerable interest. Using an $n\times n$ torus to
approximate the infinite periodic graph, we give an explicit
integral formula for the probability of a specific local
configuration at a fixed vertex, see section 6. These results follow
from a study of the zeros of the \textbf{characteristic polynomial},
or the \textbf{spectral curve}, on the unit torus $\mathbb{T}^2$.
For a more general result about the intersection of the spectral
curve with $\mathbb{T}^2$, see \cite{li}. For example, using our
method, we compute the probability that a $\{001\}$ dimer occurs and
for uniform 1-2 model, and the probability that a $\{011\}$
configuration occurs at a vertex for critical 1-2 model, see
Examples 7.2 and 7.3.

The main result of this paper can be stated in the following
theorems
\begin{theorem}
For a periodic, realizable, positive-weight vertex model on a
hexagonal lattice $G$ with period $1\times 1$, assume the
corresponding Fisher graph has positive edge weights, then the free
energy of G is
\begin{align*}
F:=\lim_{n\rightarrow\infty}\frac{1}{n^2}\log
S(G_n)=\frac{1}{8\pi^2}\iint_{|z|=1,|w|=1}\log
P(z,w)\frac{dz}{iz}\frac{dw}{iw}
\end{align*}
where $G_n$ is the quotient graph $G/(n\mathbb{Z}\times
n\mathbb{Z})$, $P(z,w)$ is the characteristic polynomial.
\end{theorem}

\begin{theorem}
Assume the periodic vertex model on hexagonal lattice is realizable
to the dimer model on a Fisher graph with positive, periodic edge
weights, and assume the entries of the corresponding base change
matrices are nonzero. Let $\lambda_n$ be the probability measure
defined for configurations on toroidal hexagonal lattice $G_n$.
Moreover, for a configuration $c$ at a vertices $v$
\begin{align*}
\lim_{n\rightarrow\infty}\lambda_n(c,v)=\sum_{d_{j}}[\prod_{i=1}^{p}w_{d_{j}}]|\mathrm{Pf}(K_{\infty}^{-1})_{V(d_{j})}|
\end{align*}
where
\begin{align*}
K_{\infty}^{-1}((u,x_1,y_1),(v,x_2,y_2))&=&\frac{1}{4\pi^2}\iint_{\mathbb{T}^2}z^{x_1-x_2}w^{y_1-y_2}\frac{\mathrm{Cof}(K(z,w))_{u,v}}{P(z,w)}\frac{dz}{iz}\frac{dw}{iw}
\end{align*}
$d_{j}$ are local dimer configurations on the gadget of the Fisher
graph corresponding to $v$. $V(d_{j})$ is the set of vertices
involved in the configuration $d_{j}$.
\end{theorem}
\noindent \textbf{Acknowledgments} The author would like to thank
Richard Kenyon for stimulating discussions. The author is also
grateful to David Wilson, B\'{e}atrice de Tili\`{e}re and Sunil
Chhita  for valuable comments.
\section{Background}

\subsection{Vertex Models}
Let $\{0,1\}^k$ denote the set of all binary sequences of length
$k$. A \textbf{vertex model} is a graph $G=(V,E)$ where we associate
to each vertex $v\in V$ a function
\begin{align*}
r_v:\{0,1\}^{deg(v)}\rightarrow\mathbb{R}^{+}
\end{align*}
$r_v$ is called the \textbf{signature} of the vertex model at vertex
$v$. We give a linear ordering on the edges adjacent to $v$, and we
fix such an ordering around each vertex once and for all. This way
the binary sequences of length $\deg(v)$ are in one-to-one
correspondence with the local configurations at $v$. Each edge
corresponds to a digit; if that edge is included in the
configuration, the corresponding digit is 1, otherwise the
corresponding digit is 0. Hence we can also consider $r_v$ as a
column vector indexed by local configurations at $v$:
\begin{align*}
r_v=\left(\begin{array}{c}r_v(0...00)\\r_v(0...01)\\r_v(0...10)\\...\\r_v(1...11)\end{array}\right)
\end{align*}

\begin{example}[signature of the vertex model at a vertex]Assume we have a
degree-2 vertex with signature
\begin{align*}
r_v=\left(\begin{array}{c}r_v(00)\\r_v(01)\\r_v(10)\\r_v(11)\end{array}\right)=\left(\begin{array}{c}\alpha\\
\beta\\ \gamma\\ \delta \end{array}\right)
\end{align*}
that means we give weights to the four different local
configurations as in Figure 1:
\begin{figure}[htbp]
\centering
\includegraphics*[276,799][336,818]{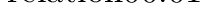}\qquad
\includegraphics*[270,794][330,815]{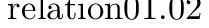}\qquad
\includegraphics*[281,797][341,814]{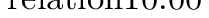}\qquad
\includegraphics*[276,796][336,813]{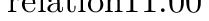}
\caption{Relation at a Vertex}
\end{figure}
\end{example}

Assume $G$ is a finite graph. We define a probability measure with
sample space the set of all configurations, $\Omega=\{0,1\}^{|E|}$.
The probability of a specific configuration $\mathcal{R}$ is
\begin{align}
\lambda(\mathcal{R})=\frac{1}{S}\prod_{v\in
V}r_v(\mathcal{R})\label{measure}.
\end{align}
The product is over all vertices. $r_v(\mathcal{R})$ is the weight
of the local configuration obtained by restricting $\mathcal{R}$ to
the vertex $v$, and $S$ is a normalizing constant called the
\textbf{partition function} for vertex models, defined to be
\begin{align*}
S=\sum_{\mathcal{R}\in\Omega}\prod_{v\in V}r_v(\mathcal{R}).
\end{align*}
The sum is over all possible configurations of $G$.

Now we consider a vertex model on a $\mathbb{Z}^2$-periodic planar
graph $G$. By this we mean that $G$ is embedded in the plane so that
translations in $\mathbb{Z}^2$ act by signature-preserving
isomorphisms of $G$. Examples of such graphs are the square and
Fisher lattices, as shown in Figure 7. Let $G_n$ be the quotient of
$G$ by the action of $n\mathbb{Z}^2$. It is a finite graph embedded
into a torus. Let $S_n$ be the partition function of the vertex
model on $G_n$. The \textbf{free energy} of the infinite periodic
vertex model $G$ is defined to be
\begin{align*}
F:=\lim_{n\rightarrow\infty}\frac{1}{n^2}\log S_n
\end{align*}

\subsection{Perfect Matching}

For more information, see \cite{ke1}. A \textbf{perfect matching},
or a \textbf{dimer cover}, of a graph is a collection of edges with
the property that each vertex is incident to exactly one edge. A
graph is \textbf{bipartite} if the vertices can be 2-colored, that
is, colored black and white so that black vertices are adjacent only
to white vertices and vice versa.

To a weighted finite graph $G=(V,E,W)$, the weight $W:
E\rightarrow\mathbb{R}^{+}$ is a function from the set of edges to
positive real numbers. We define a probability measure, called the
\textbf{Boltzmann measure} $\mu$ with sample space the set of dimer
covers. Namely, for a dimer cover $D$
\begin{align*}
\mu(D)=\frac{1}{Z}\prod_{e\in D}W(e)
\end{align*}
where the product is over all edges present in $D$, and $Z$ is a
normalizing constant called the \textbf{partition function} for
dimer models, defined to be
\begin{align*}
Z=\sum_{D}\prod_{e\in D}W(e),
\end{align*}
the sum over all dimer configurations of $G$.

If we change the weight function $W$ by multiplying the edge weights
of all edges incident to a single vertex $v$ by the same constant,
the probability measure defined above does not change. So we define
two weight functions $W,W'$ to be \textbf{gauge equivalent} if one
can be obtained from the other by a sequence of such
multiplications.

The key objects used to obtain explicit expressions for the dimer
model are \textbf{Kasteleyn matrices}. They are weighted, oriented
adjacency matrices of the graph $G$ defined as follows. A
\textbf{clockwise-odd orientation} of $G$ is an orientation of the
edges such that for each face (except the infinite face) an odd
number of edges pointing along it when traversed clockwise. For a
planar graph, such an orientation always exists \cite{ka2}. The
Kasteleyn matrix corresponding to such a graph is a
$|V(G)|\times|V(G)|$ skew-symmetric matrix $K$ defined by
\begin{align*}
K_{u,v}=\left\{\begin{array}{cc}W(uv)&{\rm if}\ u\sim v\ {\rm and}\
u\rightarrow v \\-W(uv)&{\rm if}\ u\sim v\ {\rm and}\ u\leftarrow
v\\0&{\rm else}.
\end{array}\right.
\end{align*}
It is known \cite{ka1,ka2,gt,kos} that for a planar graph with a
clock-wise odd orientation, the partition function of dimers
satisfies
\begin{align*}
Z=\sqrt{\det K}.
\end{align*}

Now let $G$ be a $\mathbb{Z}^2$-periodic planar graph. Let $G_n$ be
a quotient graph of $G$, as defined before. Let
$\gamma_{x,n}(\gamma_{y,n})$ be a path in the dual graph of $G_n$
winding once around the torus horizontally(vertically). Let
$E_H(E_V)$ be the set of edges crossed by $\gamma_{x}(\gamma_y)$. We
give a \textbf{crossing orientation} for the toroidal graph $G_n$ as
follows. We orient all the edges of $G_n$ except for those in
$E_H\cup E_V$. This is possible since no other edges are crossing.
Then we orient the edges of $E_H$ as if $E_V$ did not exist. Again
this is possible since $G-E_V$ is planar. To complete the
orientation, we also orient the edges of $E_V$ as if $E_H$ did not
exist.

For $\theta,\tau\in\{0,1\}$, let $K_n^{\theta,\tau}$ be the
Kasteleyn matrix $K_n$ in which the weights of edges in $E_H$ are
multiplied by $(-1)^{\theta}$, and those in $E_V$ are multiplied by
$(-1)^{\tau}$. It is proved in \cite{gt} that the partition function
$Z_n$ of the graph $G_n$ is
\begin{align*}
Z_n=\frac{1}{2}|\mathrm{Pf}(K_n^{00})+\mathrm{Pf}(K_n^{10})+\mathrm{Pf}(K_n^{01})-\mathrm{Pf}(K_n^{11})|.
\end{align*}

Let $E_m=\{e_1=u_1v_1,...,e_m=u_mv_m\}$ be a subset of edges of
$G_n$. Kenyon \cite{ke2} proved that the probability of these edges
occurring in a dimer configuration of $G_n$ with respect to the
Boltzmann measure $P_n$ is
\begin{align*}
P_n(e_1,...,e_m)=\frac{\prod_{i=1}^{m}W(u_iv_i)}{2Z_n}|\mathrm{Pf}(K_n^{00})_E^{c}+\mathrm{Pf}(K_n^{10})_E^{c}+\mathrm{Pf}(K_n^{01})_E^{c}-\mathrm{Pf}(K_n^{11})_E^{c}|
\end{align*}
where $E_m^c=V(G_n)\setminus\{u_1,v_1,...,u_m,v_m\}$, and
$(K_n^{\theta\tau})_{E_m^c}$ is the submatrix of $K_n^{\theta\tau}$
whose lines and columns are indexed by $E_m^c$.

The asymptotic behavior of $Z_n$ when $n$ is large is an interesting
subject. One important concept is the partition function per
fundamental domain, which is defined to be
\begin{align*}
\lim_{n\rightarrow\infty}\frac{1}{n^2}\log Z_n.
\end{align*}

Let $K_1$ be a Kasteleyn matrix for the graph $G_1$. Given any
parameters $z, w$, we construct a matrix $K(z,w)$ as follows. Let
$\gamma_{x,1}$, $\gamma_{y,1}$ be the paths introduced as above.
Multiply $K_{u,v}$ by $z$ if Pfaffian orientation on that edge is
from $u$ to $v$, otherwise multiply $K_{u,v}$ by $\frac{1}{z}$, and
similarly for $w$ on $\gamma_y$. Define the \textbf{characteristic
polynomial} $P(z,w)=\det K(z,w)$. The \textbf{spectral curve} is
defined to be the locus $P(z,w)=0$.

Gauge equivalent dimer weights give the same spectral curve. That is
because after Gauge transformation, the determinant multiplies by a
nonzero constant, and the locus of $P(z,w)$ does not change.

A formula for enlarging the fundamental domain is proved in
\cite{ckp,kos}. Let $P_n(z,w)$ be the characteristic polynomial of
$G_n$ with period $1\times 1$, and $P_1(z,w)$ be the characteristic
polynomial of $G_1$, then
\begin{align*}
P_n(z,w)=\prod_{u^n=z}\prod_{v^n=w}P_1(u,v)
\end{align*}

\subsection{Matchgates, Matchgrids}

A \textbf{matchgate} $\Gamma$ is a planar local graph including a
set $X$ of external vertices, i.e. vertices located along the
boundary of the local graph. The external vertices are ordered
anti-clockwise on the boundary. $\Gamma$ is called an odd(even)
matchgate if it has an odd(even) number of vertices.

The \textbf{signature of the matchgate} is a vector indexed by
subsets of external vertices, $\{0,1\}^{|X|}$. For a subset
$X_0\subset X$, the entry of the signature at $X_0$ is the partition
function of dimer configurations on a subgraph of the matchgate. The
subgraph is obtained from the matchgate by removing all the external
vertices in $X_0$.

\begin{example}[signature of a matchgate] Assume we have a matchgate
with external vertices 1, 2, 3, and edge weights as illustrated in
the following figure:
\begin{figure}[htbp]
\centering
\scalebox{0.7}[0.7]{\includegraphics*[210,710][377,816]{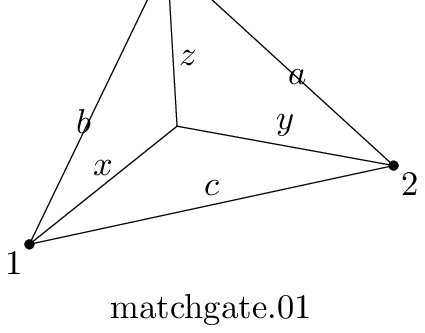}}
\caption{Matchgate}
\end{figure}

then the signature of the matchgate is
\begin{align*}
\mathfrak{m}\left(\begin{array}{c}000\\001\\010\\011\\100\\101\\110\\111\end{array}\right)=\left(\begin{array}{c}ax+by+cz\\0\\0\\x\\0\\y\\z\\0\end{array}\right)
\end{align*}
\end{example}

A \textbf{matchgrid} $M$ is a weighted planar graph consisting of a
collection of matchgates and connecting edges. Each connecting edge
has weight 1 and joins an external vertex of a matchgate with an
external vertex of another matchgate, so that every external vertex
is incident to exactly one connecting edge.

\section{Generalized Holographic Reduction}

In this section we introduce a generalized holographic algorithm to
compute the partition function of the vertex model on a finite
planar graph in terms of the partition function for perfect
matchings on a matchgrid. The idea is using a matchgate to replace
each vertex, and perform a change of basis, such that after the base
change process, the signature of a vertex becomes the signature of
the corresponding matchgate. We describe the algorithm in detail as
follows.

For a finite graph $G$, we associate to each oriented edge $e$ a
2-dimensional vector space $V_e$. To the edge with the reversed
orientation, the associated vector space is the dual space, i.e.
$V_{-e}=V_{e}^*$. Give a set basis $\{f_{e}^0,f_{e}^1\}$ for each
$V_e$, satisfying
\begin{align*}
f_{-e}^j(f_{e}^i)=\delta_{ij}.
\end{align*}
Let $v$ be a vertex with incident edges $e_{l_1},...,e_{l_{k}}$,
oriented away from $v$. The signature of the vertex model at a
vertex $v$, $r_v$, can be considered as an element in
$W_v=V_{e_{l_1}}\otimes ...\otimes V_{e_{l_k}}$. Hence $r_v$ has
representations under bases $F=\{f_e^0,f_e^1\}_{e\in E}$ and
$B=\{b_e^0,b_e^1\}_{e\in E}$ as follows
\begin{align}
r_v&=\sum_{c_{l_1},...,c_{l_k}}r_v(c_{l_1}\cdots
c_{l_k})b_{e_{l_1}}^{c_{l_1}}\otimes\cdots\otimes
b_{e_{l_k}}^{c_{l_k}}\\
&=\sum_{c_{l_1},...,c_{l_k}}r_{v,f}(c_{l_1}\cdots
c_{l_k})f_{e_{l_1}}^{c_{l_1}}\otimes\cdots\otimes
f_{e_{l_k}}^{c_{l_k}}\label{vertexsignature}
\end{align}
where $B$ are the set of standard bases for each $V_e$
\begin{align*}
b_e^0=\left(\begin{array}{c}1\\0\end{array}\right)\qquad
b_e^1=\left(\begin{array}{c}0\\1\end{array}\right)
\end{align*}
$c_{l_i}\in\{0,1\}$. $c_{l_1}\cdots c_{l_k}$ are binary sequences of
length $k$. From the definition of the signature of vertex models,
obviously $r_v(c_{l_1}\cdots c_{l_k})$ is the weight of the
configuration $c_{l_1}\cdots c_{l_k}$ at vertex $v$.

We construct a matchgrid $M$ as follows. We replace each vertex $v$
by a matchgate $\mathcal{D}_v$, such that the number of external
vertices of $\mathcal{D}_v$ is the same as the degree of $v$, and
the edges of the vertex model graph $G$ become connecting edges
joining different matchgates in the matchgrid $M$. Examples of such
replacements are illustrated in the following Figure.
\begin{figure}[htbp]
\centering
\includegraphics*[279,762][328,818]{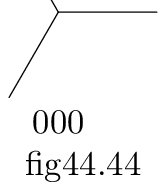}\qquad\qquad\qquad
\includegraphics*[268,733][349,819]{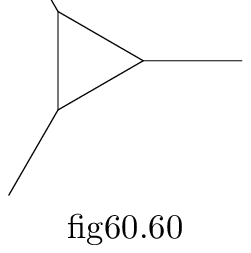}
\caption{degree-3 vertex and matchgate}
\end{figure}
\begin{figure}[htbp]
\centering
\includegraphics*[266,751][348,815]{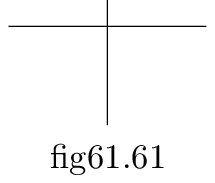}\qquad\qquad\qquad
\includegraphics*[257,725][355,817]{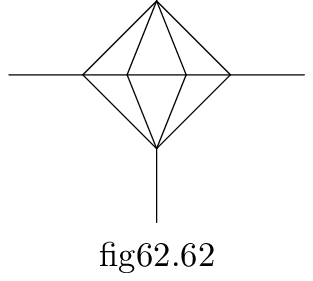}
\caption{degree-4 vertex and matchgate}
\end{figure}

If the signature $m_v$ of the matchgate $\mathcal{D}_v$ satisfies
\begin{align}
m_v=\sum_{c_{l_1}\cdots c_{l_{k}}}r_{v,f}(c_{l_1}\cdots
c_{l_k})b_{e_{l_1}}^{c_{l_1}}\otimes\cdots\otimes
b_{e_{l_k}}^{c_{l_k}},\label{matchgatesignature}
\end{align}
that is, the representation of $m_v$ under bases $B$ is the same as
the representation of $r_v$ under bases $F$, then we have the
following theorem:

\begin{theorem}Under the above base change process, the partition
function of the vertex model of $G$ is equal to the partition
function of the dimer model of $M$.
\end{theorem}
\begin{proof}There is a natural mapping $\Phi$ from $\otimes_{v\in
V}W_v$ to $\mathbb{C}$ induced by $\otimes_{e\in E}\phi_e$, where
$\phi_e$ is the natural pairing from $V_e\otimes V_e^*$ to
$\mathbb{C}$. Note that in $\otimes_{v\in V}W_v$, each $V_e$ and
$V_{-e}$ appear exactly once. Since the representation of $m_v$
under bases $B$ is the same as the representation of $r_v$ under
bases $F$, we have
\begin{align}
\Phi(\otimes_{v\in V}m_v)=\Phi(\otimes_{v\in V}r_v).\label{holant}
\end{align}
(\ref{holant}) follows from the fact that each $\phi_e: V_e\otimes
V_e^*\rightarrow \mathbb{C}$ is independent of bases as long as we
choose the dual basis for the dual vector space. However, the left
side of (\ref{holant}) is exactly the partition function of the
dimer model of $M$, while the right side of (\ref{holant}) is
exactly the partition function of the vertex model of $G$.
\end{proof}

Define the \textbf{base change matrix at edge $e$},
$T_e=\left(\begin{array}{cc}f_e^0&f_e^1\end{array}\right)$.  The
\textbf{base change matrix at vertex $v$}, $T_{v}$, acting on $W_v$
by multiplication, is defined to be
\begin{align*}
T_v=\otimes_{\{e|e\ \mathrm{is\ incident\ to\ }v,\ \mathrm{and\
oriented\ away\ from\ }v\}}T_e.
\end{align*}

In order for a vertex model problem to be reduced to dimer model
problem, one sufficient condition is that at each vertex, the
signature of the vertex model $r_v$ under the bases $F$ is the same
as the signature of the matchgate $m_v$ under the standard bases.
Namely,
\begin{align}
T_vm_v=r_v. \label{rlbcd}
\end{align}
(\ref{rlbcd}) follows directly from (\ref{vertexsignature}) and
(\ref{matchgatesignature}).

\begin{example}
Consider the graph in Figure 3 with standard dimer signature
\begin{figure}[htbp]
\centering
\includegraphics*[275,770][335,818]{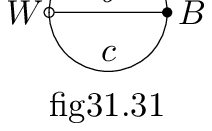}
\caption{}
\end{figure}
\begin{eqnarray*}
r_w=r_b&=&\left(\begin{array}{cccccccc}r_{000}&r_{001}&r_{010}&r_{011}&r_{100}&r_{101}&r_{110}&r_{111}\end{array}\right)^t\\\\
&=&\left(\begin{array}{cccccccc}0&1&1&0&1&0&0 &0\end{array}\right)^t
\end{eqnarray*}

Define the base change matrix on edges $a,b,c$ to be
\begin{eqnarray*}
T_a=T_b=T_c=\left(\begin{array}{cc}0&1\\1&0\end{array}\right),
\end{eqnarray*}
By definition, $T_w=T_a\otimes T_b\otimes T_c$. Note that
$T_v=(T_v^{t})^{-1}=T_b$  After the base change, we have the
standard loop signature
\begin{eqnarray*}
\tilde{r}_w=\tilde{r}_b=T_w\cdot r_w=T_b\cdot r_b\\
=\left(\begin{array}{cccccccc}0&0&0&1&0&1&1&0\end{array}\right)^t
\end{eqnarray*}

For instance, after the base change, the dimer configuration $001$
with only $c$-edge occupied becomes
\begin{eqnarray*}
&&T_w\cdot(b_0\otimes b_0\otimes b_1)\\
&&=\left(\begin{array}{cccccccc}0&0&0&0&0&0&1&0\end{array}\right),
\end{eqnarray*}
which is the configuration $110$, the loop configuration with
$a$-edge and $b$-edge occupied.
\end{example}

However, since the number of vertices in a matchgate is either even
or odd, at a vertex of degree $d$, the signature of a matchgate must
be a $2^{d-1}$ dimensional subspace of $\mathbb{C}^{2^d}$, with
those $2^{d-1}$ entries being 0. These are the entries correspond to
the partition function of dimer configurations on a subgraph of the
matchgate with an odd number of vertices, see example 2.1. This is
the \textbf{parity constraint}. As a result, by dimension count we
can see that it is not possible to use holographic algorithm to
reduce all vertex models into dimer models. To characterize the
special class of vertex models applicable to holographic reduction,
we introduce the following definition.

\begin{definition}
A network of relations on a finite graph is \textbf{realizable}, if
there exists a system of bases reducing the model to the set of
perfect matchings of a matchgrid.
\end{definition}

\begin{definition}
A network of relations is \textbf{bipartite realizable}, if it is
realizable and the corresponding matchgrid is a bipartite graph.
\end{definition}

\noindent\textbf{Remark.} The above generalizes Valiant's algorithm
\cite{val} in the sense that our basis can be different from edge to
edge. As a consequence, our approach results in an enlargement of
the dimension of realizable submanifold, which will be shown in the
next section.

\subsection{Realizability} We are interested in periodic vertex models
on the honeycomb lattice with period $n\times n$. The quotient graph
can be embedded on a torus $\mathbb{T}^2=S^1\times S^1$.  We
classify all the edges into $a$-type, $b$-type and $c$-type
according to their direction, and assume $b$-type and $c$-type edges
have the same direction as the two basic homology cycles $(1,0)$ and
$(0,1)$ of torus, respectively, see Figure 6.
\begin{figure}[htbp]
\centering
\scalebox{0.8}[0.8]{\includegraphics*[176,652][430,820]{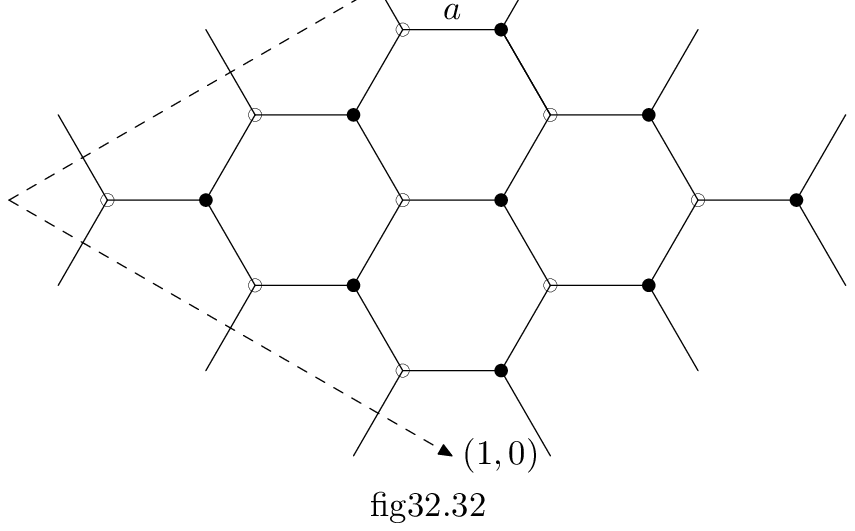}}
\caption{Periodic Honeycomb Lattice}
\end{figure}

Assume the vertex model is realizable, then each corresponding
matchgate may have either an even or an odd number of vertices. By
enlarging the fundamental domain, we can always assume that there
are an even number of matchgates with an even number of vertices.
Then by permuting rows of matrices on a finite number of edges, we
can always have all the matchgates having an odd number of vertices.
For example, assume we have a pair of adjacent matchgates with
matrix on the connecting edge $e$,
$T_e=\left(\begin{array}{c}t_0\\t_1\end{array}\right)$($T_e$ here is
actually the inverse of the base change matrix defined before ), as
illustrated in Figure 7.
\begin{figure}[htbp]
\centering
\scalebox{0.7}[0.7]{\includegraphics*[210,710][400,800]{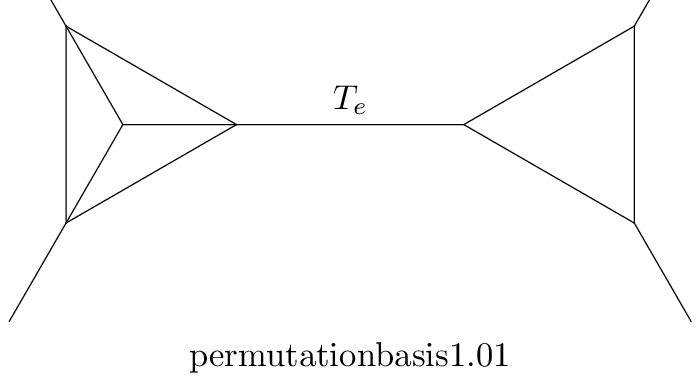}}
\caption{Permutation of Basis Vectors 1}
\end{figure}
Then if we assume
$\tilde{T}_e=\left(\begin{array}{c}t_1\\t_0\end{array}\right)$
obtained by permutating two basis vectors of $T_e$, we actually
interchange the roles of 0 and 1 at the corresponding digit of the
binary sequences as indices of signatures of the matchgate. Namely,
assume $m_v=T_vr_v$, and $\tilde{m}_v=\tilde{T}_vr_v$, where
$\tilde{T}_v=\tilde{T}_a\otimes T_b\otimes T_c$. Then for any binary
sequence $c_1c_2c_3$, the entry ${m_v}_{\{c_1c_2c_3\}}$ is the same
as $\tilde{m_v}_{\{(1-c_1)c_2c_3\}}$, according to equation
(\ref{rlbcd}). If originally we have an even matchgate at $v$, by
parity constraint, the $001,010,100,111$ entries of $m_v$ are zero.
After the permutation of basis vectors, $\tilde{m}_v$ will have
$101,110,000,011$ entries to be zero. Hence $\tilde{m}_v$ has to be
an odd matchgate, as illustrate in Figure 8.
\begin{figure}[htbp]
\centering
\scalebox{0.7}[0.7]{\includegraphics*[210,710][400,800]{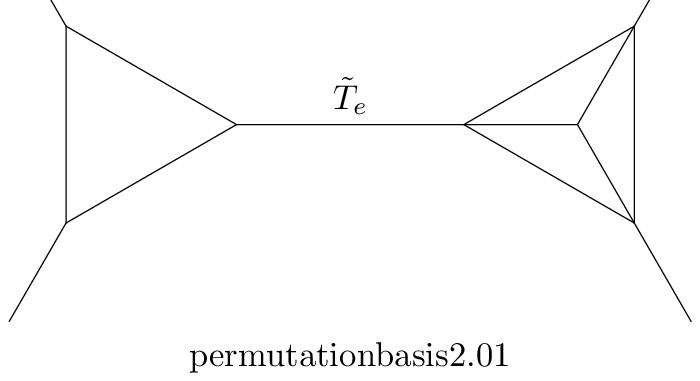}}
\caption{Permutation of Basis Vectors 2}
\end{figure}
By finitely many times of such permutations, we can move two even
matchgates to adjacent positions. Then we permutate the basis on the
connecting edge, we can decrease the number of even matchgates by 2.
Since we have an even number of even matchgates in total, if we
repeat the process, in the end, all matchgates will be odd.

\begin{assumption}From now on, we make the following assumption:
\begin{itemize}
\item All entries of base change matrices are nonzero.
\item All entries of the matchgate signature are nonzero.
\end{itemize}
\end{assumption}

Since the honeycomb lattice is a bipartite graph, we can color all
the vertices in black and white such that black vertices are
adjacent only to white vertices and vice versa. Assume vertex
signatures at black and white vertices are as follows:
\begin{align}
r_w^{ij}&=\left(\begin{array}{cccccccc}x_{000}^{ij}&x_{001}^{ij}&x_{010}^{ij}&x_{011}^{ij}&x_{100}^{ij}&x_{101}^{ij}&x_{110}^{ij}&x_{111}^{ij}\end{array}\right)^t\label{relationsignaturex}\\
r_b^{ij}&=\left(\begin{array}{cccccccc}y_{000}^{ij}&y_{001}^{ij}&y_{010}^{ij}&y_{011}^{ij}&y_{100}^{ij}&y_{101}^{ij}&y_{110}^{ij}&y_{111}^{ij}\end{array}\right)^t\label{relationsignaturey}
\end{align}
where $(i,j),1\leq i\leq n,1\leq j\leq n$ is the row and column
index of white and black vertices. Assume $000=1$, then entries of
signatures are indexed from 1 to 8. Associate to an $a$-edge,
$b$-edge and $c$-edge incident to a white or black vertices, a basis
\begin{eqnarray*}
T_a^{(i,j,k)}=\left(\begin{array}{cc}n_{01}^{(i,j,k)}&p_{01}^{(i,j,k)}\\n_{11}^{(i,j,k)}&p_{11}^{(i,j,k)}\end{array}\right)=\left(\begin{array}{cc}\textbf{n}_1^{(i,j,k)}& \textbf{p}_1^{(i,j,k)}\end{array}\right)\\
T_b^{(i,j,k)}=\left(\begin{array}{cc}n_{02}^{(i,j,k)}&p_{02}^{(i,j,k)}\\n_{12}^{(i,j,k)}&p_{12}^{(i,j,k)}\end{array}\right)=\left(\begin{array}{cc}\textbf{n}_2^{(i,j,k)}& \textbf{p}_2^{(i,j,k)}\end{array}\right)\\
T_c^{(i,j,k)}=\left(\begin{array}{cc}n_{03}^{(i,j,k)}&p_{03}^{(i,j,k)}\\n_{13}^{(i,j,k)}&p_{13}^{(i,j,k)}\end{array}\right)=\left(\begin{array}{cc}\textbf{n}_3^{(i,j,k)}& \textbf{p}_3^{(i,j,k)}\end{array}\right)\\
\end{eqnarray*}
where $k=1$ for white vertices and $k=0$ for black vertices.

By the realizability equation (\ref{rlbcd}), and the fact that all
matchgates are odd, the 1st, 4th, 6th, 7th entries of the matchgate
signatures $m_v$ are 0, so we have a system of 4 algebraic equations
at each vertex $v$. For example, at each black vertex, we have
\begin{align}\label{blackrealizability}
m_{i,j,0}=T_a^{i,j,0}\otimes T_b^{i,j,0}\otimes T_c^{i,j,0}\cdot
r_b^{ij}
\end{align}
and the fact that the 1st entry of $m_{i,j,0}$ is 0 gives the
following equation
\begin{align*}
 n_{01}n_{02}n_{03}y_1+n_{01}n_{02}p_{03}y_2+n_{01}p_{02}n_{03}y_3+n_{01}p_{02}p_{03}y_4\\
 +p_{01}n_{02}n_{03}y_5+p_{01}n_{02}p_{03}y_6+p_{01}p_{02}n_{03}y_7+p_{01}p_{02}p_{03}y_8=0
\end{align*}
That the 4th, 6th, 7th entries of $m_{i,j,0}$ are 0 give three other
similar equations. Similarly at each white vertex, we have
\begin{equation}\label{whiterealizability}
m_{i,j,1}=[(T_a^{i,j,1}\otimes T_b^{i,j,1}\otimes
T_c^{i,j,1})^t]^{-1}\cdot r_w^{ij}
\end{equation}
the same process gives a system of 4 equations at the white vertex.

Let
$a_{lm}^{ij}=\frac{n_{lm}^{(i,j,1)}}{p_{lm}^{(i,j,1)}}$,$b_{lm}^{ij}=\frac{n_{lm}^{(i,j,0)}}{p_{lm}^{(i,j,0)}}$,
$l\in\{0,1\}$,$m\in\{1,2,3\}$. Then the equations we get are linear
with respect to $a_{01}^{ij}$, $a_{11}^{ij}$, $b_{01}^{ij}$,
$b_{11}^{ij}$, which can be solved explicitly.
\begin{eqnarray}
a_{01}^{ij}=\frac{r^{ij}\cdot u^{ij}}{r^{ij}\cdot v^{ij}}=\frac{q^{ij}\cdot u^{ij}}{q^{ij}\cdot v^{ij}}\label{expbe1}\\
a_{11}^{ij}=\frac{s^{ij}\cdot u^{ij}}{s^{ij}\cdot v^{ij}}=\frac{p^{ij}\cdot u^{ij}}{p^{ij}\cdot v^{ij}}\\
b_{01}^{ij}=-\frac{\xi^{ij}\cdot t^{ij}}{\xi^{ij}\cdot
w^{ij}}=-\frac{\kappa^{ij}\cdot t^{ij}}{\kappa^{ij}\cdot w^{ij}}\\
b_{11}^{ij}=-\frac{\sigma^{ij}\cdot t^{ij}}{\sigma^{ij}\cdot
w^{ij}}=-\frac{\lambda^{ij}\cdot t^{ij}}{\lambda^{ij}\cdot
w^{ij}}\label{expbe2}
\end{eqnarray}
where
\begin{eqnarray}\label{variable}
\left\{\begin{array}{ll}p^{ij}=\left(\begin{array}{cccc}a_{02}^{ij}a_{03}^{ij}&
a_{02}^{ij}& a_{03}^{ij}& 1\end{array}\right)&
q^{ij}=\left(\begin{array}{cccc}a_{02}^{ij}a_{13}^{ij}& a_{02}^{ij}&
a_{13}^{ij}& 1\end{array}\right)\\
r^{ij}=\left(\begin{array}{cccc}a_{12}^{ij}a_{03}^{ij}& a_{12}^{ij}&
a_{03}^{ij}& 1\end{array}\right)&
s^{ij}=\left(\begin{array}{cccc}a_{12}^{ij}a_{13}^{ij}& a_{12}^{ij}&
a_{13}^{ij}& 1\end{array}\right)\\
\xi^{ij}=\left(\begin{array}{cccc}b_{02}^{ij}b_{03}^{ij}&
b_{02}^{ij}& b_{03}^{ij}& 1\end{array}\right)&
\sigma^{ij}=\left(\begin{array}{cccc}b_{02}^{ij}b_{13}^{ij}&
b_{02}^{ij}& b_{13}^{ij}& 1\end{array}\right)\\
\lambda^{ij}=\left(\begin{array}{cccc}b_{12}^{ij}b_{03}^{ij}&
b_{12}^{ij}& b_{03}^{ij}& 1\end{array}\right)&
\kappa^{ij}=\left(\begin{array}{cccc}b_{12}^{ij}b_{13}^{ij}&
b_{12}^{ij}& b_{13}^{ij}& 1\end{array}\right)\\
u^{ij}=\left(\begin{array}{cccc}x_4^{ij}&-x_3^{ij}&-x_2^{ij}&
x_1^{ij}\end{array}\right)&
v^{ij}=\left(\begin{array}{cccc}x_8^{ij}&-x_7^{ij}&-x_6^{ij}&
x_5^{ij}\end{array}\right)\\
w^{ij}=\left(\begin{array}{cccc}y_1^{ij}&y_2^{ij}&y_3^{ij}&
y_4^{ij}\end{array}\right)&
t^{ij}=\left(\begin{array}{cccc}y_5^{ij}&y_6^{ij}&y_7^{ij}&
y_8^{ij}\end{array}\right)\end{array}\right.
\end{eqnarray}

Since $a_{0m}\neq a_{1m}$ and $b_{0m}\neq b_{1m}$, after clearing
denominators and some reducing, for any $(i,j)$, equations
(\ref{expbe1})-(\ref{expbe2}) are equivalent to
\begin{eqnarray}
2y_2y_8-2y_6y_4+(b_{13}+b_{03})(y_1y_8+y_2y_7-y_5y_4-y_6y_3)+2(y_1y_7-y_5y_3)b_{13}b_{03}=0\label{s1}\\
-2x_1x_7+2x_5x_3+(a_{13}+a_{03})(x_2x_7+x_1x_8-x_4x_5-x_6x_3)+2(x_6x_4-x_2x_8)a_{13}a_{03}=0\label{s2}\\
2y_8y_3-2y_4y_7+(b_{12}+b_{02})(y_1y_8+y_3y_6-y_4y_5-y_7y_2)+2b_{12}b_{02}(y_6y_1-y_2y_5)=0\label{s3}\\
2x_5x_2-2x_1x_6+(x_1x_8-x_2x_7-x_5x_4+x_3x_6)(a_{12}+a_{02})+2(x_7x_4-x_3x_8)a_{02}a_{12}=0\label{s4}
\end{eqnarray}

If we solve $a_{02},a_{12},b_{02},b_{12}$ explicitly, a similar
process yields
\begin{eqnarray}
2x_3x_2-2x_1x_4+(x_5x_4+x_1x_8-x_2x_7-x_3x_6)(a_{01}+a_{11})+2(x_6x_7-x_5x_8)a_{01}a_{11}=0\label{s5}\\
2y_6y_7-2y_5y_8+(y_3y_6+y_2y_7-y_1y_8-y_4y_5)(b_{01}+b_{11})+2(y_2y_3-y_1y_4)b_{01}b_{11}=0\label{s6}
\end{eqnarray}

We get two equations per edge, one involving the a-variables, the
other involving the b-variables. But a-variables and b-variables are
actually the same thing for each single edge. From equation
(\ref{s1}),(\ref{s2}), we can solve basis entries
$a_{l,3}^{i,j}=b_{l,3}^{i,j-1}$; from equation
(\ref{s3}),(\ref{s4}), we can solve basis entries
$a_{l,2}^{i,j}=b_{l,2}^{i-1,j}$. Finally from
(\ref{expbe1})-(\ref{expbe2}), we can solve $a_{l,1}$ and $b_{l,1}$,
the only constraint left will be $a_{l1}^{i,j}=b_{l,1}^{i,j}$, which
is two polynomial equations with respect to relation signature at
each a-edge. Together with Theorem 2 in appendix, we have the
following theorem

\begin{theorem}
Under assumption 3.5, the realizable signatures on the $n\times n$
periodic honeycomb lattice form a $14n^2$ dimensional submanifold of
the $16n^2$ dimensional manifold of all positive signatures; the
bipartite realizable signatures on the $n\times n$ periodic
honeycomb lattice form a $12n^2$ dimensional submanifold of the
$16n^2$ dimensional manifold of all positive signatures.
\end{theorem}

For the exact realizability condition, see the appendix.

Under assumption 3.5, the weight \{111\} is non-vanishing at each
matchgate. Since the probability measure will not change if all
entries of the signature at one vertex are multiplied by a constant,
we can assume at each matchgate, the weight of $\{111\}$ is 1.
Therefore, we have

\begin{theorem}
A realizable vertex model on a finite hexagonal lattice $G$ can
always be transformed to dimers on $M$, a Fisher graph as shown in
Figure 9, with the partition function of the vertex model of $G$
equal to the partition function of the dimer model of $M$, up to
multiplication of a constant.
\end{theorem}

\begin{figure}[htbp]
  \centering
\includegraphics*[115,583][500,825]{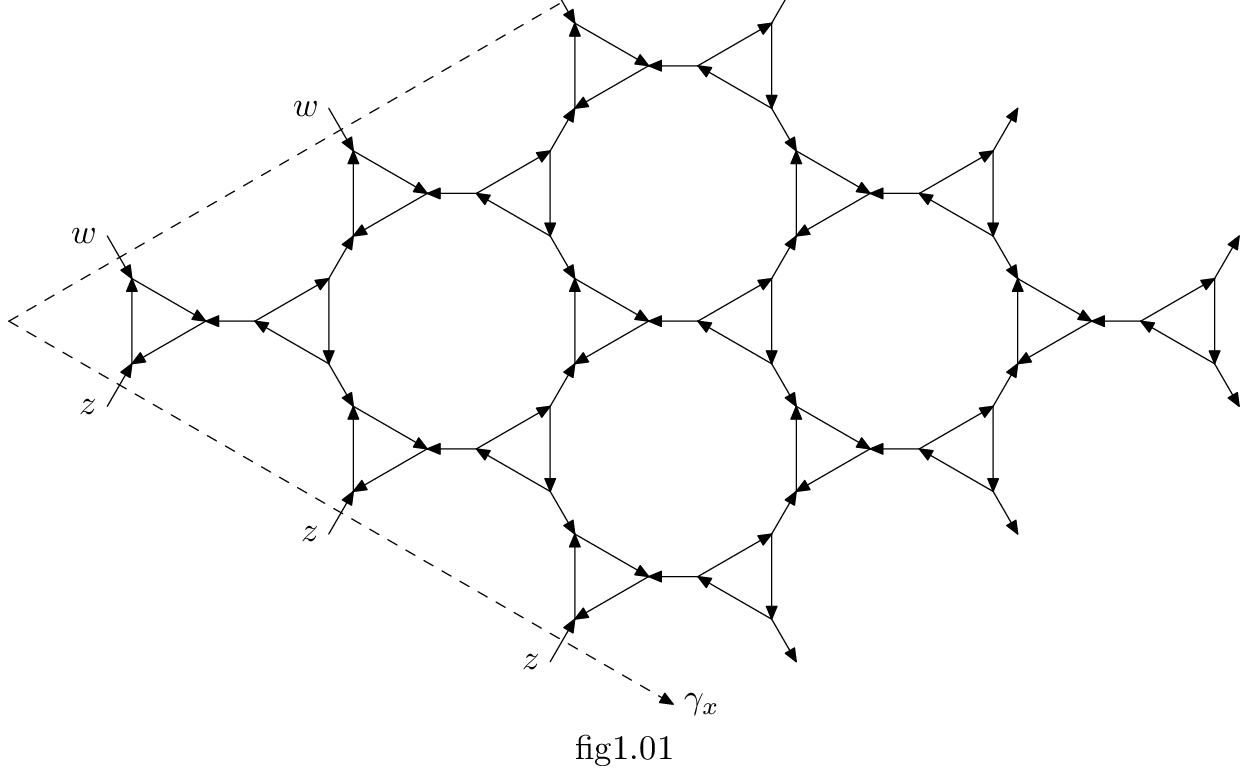}
   \caption{Matchgrid with $3\times$ 3 Fundamental Domain }
\end{figure}

It is possible to construct a matchgrid with different weights to
produce the same vertex model. Since the holographic reduction is an
invertible process, by which we mean that because the base change
matrices are nonsingular, we can always achieve the matchgate
signature from the vertex signature and vice versa, there is an
equivalence relation on dimer models producing the same vertex
model.

\begin{definition}
A vertex model is holographically equivalent to a dimer model on a
matchgrid, if it can be reduced to the dimer model on the matchgrid
using the holographic algorithm, in such a way that partition
function of the vertex model corresponds to perfect matchings of the
matchgrid. Two matchgrids are holographically equivalent, if they
are holographically equivalent to the same vertex model.
\end{definition}

\begin{proposition}
Holographically equivalent matchgrids give rise to gauge equivalent
dimer models. Therefore, they have the same probability measure.
\end{proposition}
\begin{proof}See the appendix.
\end{proof}

\subsection{Orthogonal Realizability }

\begin{definition}
A vertex model is orthogonally realizable if it is realizable by an
orthonormal base change matrix on each edge.
\end{definition}

Consider a single vertex. To the incident edges of the vertex,
associate an matrices $U_1,U_2,U_3$. Without loss of generality, we
can assume $U_1,U_2,U_3\in SO(2)$. In fact, if $\det U_i=-1$, we
multiply the first row of $U_i$ by $-1$. The new signature of the
matchgate will be multiplied by $-1$. This change does not violate
the parity constraint.

Assume
$U_i=\left(\begin{array}{cc}\cos\alpha_i&\sin\alpha_i\\-\sin\alpha_i&\cos\alpha_i\end{array}\right)$.
Assume $U=\otimes_{i=1}^{3}U_{i}$, then $U\in SO(8)$, each term of
which is product of three of $\cos{\alpha_i},\sin{\alpha_i}$.
Moreover, the eigenvalues are
$e^{(\pm\alpha_1\pm\alpha_2\pm\alpha_3)\sqrt{-1}}$.  Using
trigonometric identities, each entry is a linear combination of
$\cos(\pm\alpha_1\pm\alpha_2\pm\alpha_3)$ and
$\sin(\pm\alpha_1\pm\alpha_2\pm\alpha_3)$. If we further define
\begin{center}
\begin{align*}
g=\left(\begin{array}{cccccccc}g_1&g_2&g_3&g_4&g_5&g_6&g_7&g_8\end{array}\right)'\\
h=\left(\begin{array}{cccccccc}0&h_2&h_3&0&h_5&0&0&h_8\end{array}\right)'\\
\gamma=\alpha_1+\alpha_2-\alpha_3;\
\psi=\alpha_1+\alpha_3-\alpha_2;\
\varphi=\alpha_2+\alpha_3-\alpha_1\\
j_1=g_4+g_6+g_7-g_1;\ j_2=g_1+g_6+g_7-g_4;\\
j_3=g_1+g_4+g_7-g_6;\ j_4=g_1+g_4+g_6-g_7;\\
j_5=g_3+g_5+g_8-g_2;\ j_6=g_2+g_5+g_8-g_3;\\
j_7=g_2+g_3+g_8-g_5;\ j_8=g_2+g_3+g_5-g_8
\end{align*}
\end{center}
\begin{align*}
P=j_2\cos(\varphi)+j_7\sin(\varphi);\
Q=j_3\cos(\psi)+j_6\sin(\psi);\
R=j_4\cos(\gamma)+j_5\sin(\gamma)\\
K=j_1\cos(\varphi+\psi+\gamma)-j_8\sin(\varphi+\psi+\gamma)\\
H=j_7\cos(\varphi)-j_2\sin(\varphi);\
L=j_6\cos(\psi)-j_3\sin(\psi);\ M=j_5\cos(\gamma)-j_4\sin(\gamma)\\
N=j_8\cos(\varphi+\psi+\gamma)+j_1\sin(\varphi+\psi+\gamma)
\end{align*}
then if $Ug=h$, we have
\begin{eqnarray}
0=P=Q=R=K\label{orthogonalrealizability}\\
4h_2=H+L-M+N\\
4h_3=H-L+M+N\\
4h_5=N+M+L-H\\
4h_8=H+L+M-N
\end{eqnarray}

By (\ref{orthogonalrealizability}), we have
$\tan\varphi=-\frac{j_2}{j_7}$,$\tan\psi=-\frac{j_3}{j_6}$,$\tan\gamma=-\frac{j_4}{j_5}$,$\tan(\varphi+\psi+\gamma)=\frac{j_1}{j_8}$.
If we define
\begin{align*}
t(u,v)&=&\frac{u+v}{1-uv}\\
tt(a,b,c,d)&=&\frac{a+b+c+d-abc-abd-acd-bcd}{1-ab-ac-ad-bc-bd-cd+abcd}
\end{align*} Then the following theorem holds:

\begin{theorem}
A vertex model on a periodic honeycomb lattice with period $n\times
n$ is orthogonally realizable if and only if its signatures satisfy
the following system of equations
\begin{align*}
tt(-\frac{z_7^{ijk}}{z_2^{ijk}},-\frac{z_6^{ijk}}{z_3^{ijk}},-\frac{z_4^{ijk}}{z_5^{ijk}},-\frac{z_1^{ijk}}{z_8^{ijk}})=0\qquad\forall i,j,k\\
\left.\begin{array}{c}t(-\frac{z_6^{ij0}}{z_3^{ij0}},-\frac{z_4^{ij0}}{z_5^{ij0}})=t(-\frac{z_6^{ij1}}{z_3^{ij1}},-\frac{z_4^{ij1}}{z_5^{ij1}})\\
t(-\frac{z_3^{ij-1,0}}{z_6^{ij-1,0}},-\frac{z_2^{ij-1,0}}{z_7^{ij-1,0}})=t(-\frac{z_3^{ij1}}{z_6^{ij1}},-\frac{z_2^{ij,1}}{z_7^{ij1}})\\
t(-\frac{z_7^{i-1,j0}}{z_2^{i-1,j0}},-\frac{z_4^{i-1,j0}}{z_5^{i-1,j0}})=t(-\frac{z_7^{ij1}}{z_2^{ij1}},-\frac{z_4^{ij1}}{z_5^{ij1}})\end{array}\right\}\rightarrow\forall
i,j
\end{align*}
where
$z_1^{ij1}=x_4+x_6+x_7-x_1,z_2^{ij1}=x_3+x_5+x_8-x_2,z_3^{ij1}=x_2+x_5+x_8-x_3,z_4^{ij1}=x_1+x_6+x_7-x_4,z_5^{ij1}=x_2+x_3+x_8-x_5,z_6^{ij1}=x_1+x_4+x_7-x_6,z_7^{ij1}=x_1+x_4+x_6-x_7,z_8^{ij1}=x_2+x_3+x_5-x_8$,
and the same relation for $z_l^{ij0}$ and $y_1,...,y_8$. $x's$ and
$y's$ are defined by ($\ref{relationsignaturex}$) and
$(\ref{relationsignaturey})$.
\end{theorem}

It is trivial to verify these equations in any given situation.

\begin{definition}
A vertex model on a periodic honeycomb lattice with period $n\times
n$ is positively orthogonally realizable if it is orthogonal
realizable and for each vertex $(i,j,k)$, there exists angles
$\varphi^{ijk},\psi^{ijk},\gamma^{ijk}$, such that
\begin{align*}
\sin\varphi=\frac{z_4^{ijk}}{\sqrt{(z_4^{ijk})^2+(z_5^{ijk})^2}};\\
\sin\psi=\frac{z_6^{ijk}}{\sqrt{(z_3^{ijk})^2+(z_6^{ijk})^2}};\\
\sin\gamma=\frac{z_7^{ijk}}{\sqrt{(z_7^{ijk})^2+(z_2^{ijk})^2}}\\
\sin(-\gamma-\varphi-\psi)=\frac{z_1^{ijk}}{\sqrt{(z_1^{ijk})^2+(z_8^{ijk})^2}}\\
\end{align*}
\end{definition}

\begin{proposition}
If a vertex model on a bi-periodic hexagonal lattice is orthogonally
realizable, then the corresponding dimer configuration has positive
edge weights.
\end{proposition}
\begin{proof}
 Under the assumption that the vertex model have nonnegative signature, we
have
\begin{align*}
z_1\leq z_4+z_6+z_7;\  z_4\leq z_1+z_6+z_7;\ z_6\leq z_1+z_4+z_7;\ z_7\leq z_4+z_6+z_7\\
z_2\leq z_3+z_5+z_8;\ z_3\leq z_2+z_5+z_8;\ z_5\leq z_2+z_3+z_8;\
z_8\leq z_2+z_3+z_5
\end{align*}
at all vertices. Since at least three of $z_1,z_4, z_6, z_7$ are
nonnegative, similarly for $z_2,z_3,z_5,z_8$, if we take absolute
value for all $z_i$, the above inequalities also hold. therefore
\begin{align*}
\sqrt{z_4^2+z_5^2}+\sqrt{z_6^2+z_3^2}+\sqrt{z_2^2+z_7^2}-\sqrt{z_1^2+z_8^2}\geq 0\\
\sqrt{z_4^2+z_5^2}+\sqrt{z_6^2+z_3^2}+\sqrt{z_1^2+z_8^2}-\sqrt{z_2^2+z_7^2}\geq 0\\
\sqrt{z_4^2+z_5^2}+\sqrt{z_1^2+z_8^2}+\sqrt{z_2^2+z_7^2}-\sqrt{z_3^2+z_6^2}\geq 0\\
\sqrt{z_1^2+z_8^2}+\sqrt{z_6^2+z_3^2}+\sqrt{z_2^2+z_7^2}-\sqrt{z_4^2+z_5^2}\geq
0
\end{align*}
Under the assumption that both the vertex model are positively
orthogonal realizable, the left side of the above inequalities are
exactly edge weights of corresponding matchgates.
\end{proof}

\begin{theorem}
If all matchgates have the same signature, for a generic choice of
signature, a vertex model on a hexagonal lattice is realizable if
and only if it is orthogonal realizable.
\end{theorem}

\begin{proof}
Obviously orthogonal realizability implies realizability, we only
need to show that realizability implies orthogonal realizability.

Assume
$r=\left(\begin{array}{cccccccc}x_1&x_2&x_3&x_4&x_5&x_6&x_7&x_8\end{array}\right)$
is a realizable signature for the vertex model. By lemma 4.1, we can
assume the corresponding bases on all edges have real entries.
Consider the singular value decomposition for the base change matrix
on each edge. Without loss of generality, assume
\begin{align*}
T_i=U_i\left(\begin{array}{cc}1&0\\0&\lambda_i\end{array}\right)
V_i\ \ i=1,2,3
\end{align*}
where $U_i,V_i\in O(2)$, and $\lambda_i$ is a nonnegative real
number. Assume
\begin{align*}
v=(\otimes_{i=1}^{3}
V_{i})r^{t}=\left(\begin{array}{cccccccc}v_1&v_2&v_3&v_4&v_5&v_6&v_7&v_8\end{array}\right)^{t}
\end{align*}
Then we can consider
$(\otimes_{i=1}^{3}\left(\begin{array}{cc}1&0\\0&
\lambda_i\end{array}\right))v$ and $(\otimes_{i=1}^{3}\left(\begin{array}{cc}1&0\\
0&\frac{1}{\lambda_i}\end{array}\right))v$ to be signatures of black
and white vertices that are orthogonal realizable. Then by the
conditions of orthogonal realizability, we have
\begin{eqnarray}
\frac{v_1+\lambda_1\lambda_3 v_6+\lambda_1\lambda_2
v_7-\lambda_2\lambda_3 v_4}{\lambda_3 v_2+\lambda_2
v_3+\lambda_1\lambda_2\lambda_3 v_8-\lambda_1
v_5}=\frac{v_1+\frac{1}{\lambda_1\lambda_3}v_6+\frac{1}{\lambda_1\lambda_2}v_7-\frac{1}{\lambda_2\lambda_3}v_4}{\frac{1}{\lambda_3}v_2+\frac{1}{\lambda_2}v_3+\frac{1}{\lambda_1\lambda_2\lambda_3}v_8-\frac{1}{\lambda_1}v_5}\label{svd1}\\
\frac{v_1+\lambda_1\lambda_3 v_6-\lambda_1\lambda_2
v_7+\lambda_2\lambda_3 v_4}{-\lambda_3 v_2+\lambda_2
v_3+\lambda_1\lambda_2\lambda_3 v_8+\lambda_1
v_5}=\frac{v_1+\frac{1}{\lambda_1\lambda_3}v_6-\frac{1}{\lambda_1\lambda_2}v_7+\frac{1}{\lambda_2\lambda_3}v_4}{-\frac{1}{\lambda_3}v_2+\frac{1}{\lambda_2}v_3+\frac{1}{\lambda_1\lambda_2\lambda_3}v_8+\frac{1}{\lambda_1}v_5}\label{svd2}\\
\frac{v_1-\lambda_1\lambda_3 v_6+\lambda_1\lambda_2
v_7+\lambda_2\lambda_3 v_4}{\lambda_3 v_2-\lambda_2
v_3+\lambda_1\lambda_2\lambda_3 v_8+\lambda_1
v_5}=\frac{v_1-\frac{1}{\lambda_1\lambda_3}v_6+\frac{1}{\lambda_1\lambda_2}v_7+\frac{1}{\lambda_2\lambda_3}v_4}{\frac{1}{\lambda_3}v_2-\frac{1}{\lambda_2}v_3+\frac{1}{\lambda_1\lambda_2\lambda_3}v_8+\frac{1}{\lambda_1}v_5}\label{svd3}\\
\frac{-v_1+\lambda_1\lambda_3 v_6+\lambda_1\lambda_2
v_7+\lambda_2\lambda_3 v_4}{\lambda_3 v_2+\lambda_2
v_3-\lambda_1\lambda_2\lambda_3 v_8+\lambda_1
v_5}=\frac{-v_1+\frac{1}{\lambda_1\lambda_3}v_6+\frac{1}{\lambda_1\lambda_2}v_7+\frac{1}{\lambda_2\lambda_3}v_4}{\frac{1}{\lambda_3}v_2+\frac{1}{\lambda_2}v_3-\frac{1}{\lambda_1\lambda_2\lambda_3}v_8+\frac{1}{\lambda_1}v_5}\label{svd4}
\end{eqnarray}
If we transform fractions into polynomials, and consider
(\ref{svd1})$+$(\ref{svd2})$-$(\ref{svd3})$-$(\ref{svd4}),
(\ref{svd1})$+$(\ref{svd3})$-$(\ref{svd2})$-$(\ref{svd4}), and
(\ref{svd1})$+$(\ref{svd4})$-$(\ref{svd2})$-$(\ref{svd3}), then
\begin{align*}
(\lambda_1+1)(\lambda_1-1)(v_4v_8+v_1v_5+v_3v_7+v_2v_6)=0\\
(\lambda_2+1)(\lambda_2-1)(v_3v_4+v_5v_6+v_1v_2+v_7v_8)=0\\
(\lambda_3+1)(\lambda_3-1)(v_4v_2+v_7v_5+v_3v_1+v_8v_6)=0
\end{align*}
Then for a generic choice of signature, we must have
$\lambda_1=\lambda_2=\lambda_3=1$.
\end{proof}

\section{Characteristic Polynomial}

Assume all entries of the vertex signatures are strictly positive.
In this section we prove some interesting properties of the
characteristic polynomial.
\begin{lemma}
For realizable vertex model with positive signature and period
$1\times 1$, there exists a realization of base change over
$\mathbf{GL}_2(\mathbb{R})$, (i.~e.~ one can take base change
matrices to be real), with the property that, at each edge,
$\frac{n_0n_1}{p_0p_1}< 0$.
\end{lemma}
\begin{proof}
For $1\times 1$ fundamental domain, $a_{lk}=b_{lk},\ \forall\ l,k$.
By (\ref{expbe1})-(\ref{expbe2}), we have
\begin{align*}
a_{01}=-\frac{p\cdot t}{p\cdot w}=-\frac{s\cdot t}{s\cdot
w}=\frac{r\cdot u}{r\cdot v}=\frac{q\cdot u}{q\cdot v}\\
a_{11}=-\frac{q\cdot t}{q\cdot w}=-\frac{r\cdot t}{r\cdot
w}=\frac{s\cdot u}{s\cdot v}=\frac{p\cdot u}{p\cdot v},
\end{align*}
From which we derive
\begin{eqnarray}
(p\cdot t)(r\cdot v)+(p\cdot w)(r\cdot u)-(r\cdot t)(p\cdot
v)-(r\cdot w)(p\cdot u)=0\label{1by1a}\\ (s\cdot t)(q\cdot
v)+(s\cdot w)(q\cdot u)-(q\cdot t)(s\cdot v)-(q\cdot w)(s\cdot
u)=0\\ (p\cdot t)(q\cdot v)+(p\cdot w)(q\cdot u)-(q\cdot t)(p\cdot
v)-(q\cdot w)(p\cdot
u)=0\\
(s\cdot t)(r\cdot v)+(s\cdot w)(r\cdot u)-(r\cdot t)(s\cdot
v)-(r\cdot w)(s\cdot u)=0\label{1by1b}.
\end{eqnarray}
Since each base change matrix is invertible, we have $a_{0i}\neq
a_{1i}$. Plugging (\ref{variable}) into (\ref{1by1a})-(\ref{1by1b}),
and factor out $(a_{02}-a_{12})$, or $(a_{03}-a_{13})$, we obtain
that $a_{03}$, $a_{13}$ are two roots of quadratic polynomial
\begin{align*}
(x_6y_5+y_7x_8+y_1x_2+y_3x_4)z^2+(-y_7x_7-y_3x_3+y_2x_2+y_4x_4+x_6y_6-x_5y_5+y_8x_8-y_1x_1)z\\
-y_6x_5-x_3y_4-y_2x_1-x_7y_8=0,
\end{align*}
and $a_{02}$, $a_{12}$ are two roots of quadratic polynomial
\begin{align*}
(x_3y_1+y_6x_8+y_2x_4+y_5x_7)z^2+(y_8x_8-x_6y_6-x_5y_5+y_7x_7-y_2x_2+y_3x_3+y_4x_4-y_1x_1)z\\
-y_8x_6-y_7x_5-x_1y_3-y_4x_2=0
\end{align*}
Under the assumption that $x_i$,$y_j$ are positive, these
polynomials have real roots, since the products of two roots are
always negative. Then the lemma follows.
\end{proof}

By definition, the characteristic polynomial $P(z,w)$ is the
determinant of an $6\times 6$ matrix $K(z,w)$ whose rows and columns
are indexed by the $6$ vertices in a $1\times 1$ fundamental domain,
see Figure 8:
\begin{figure}[htbp]
  \centering
\includegraphics*[190,655][417,817]{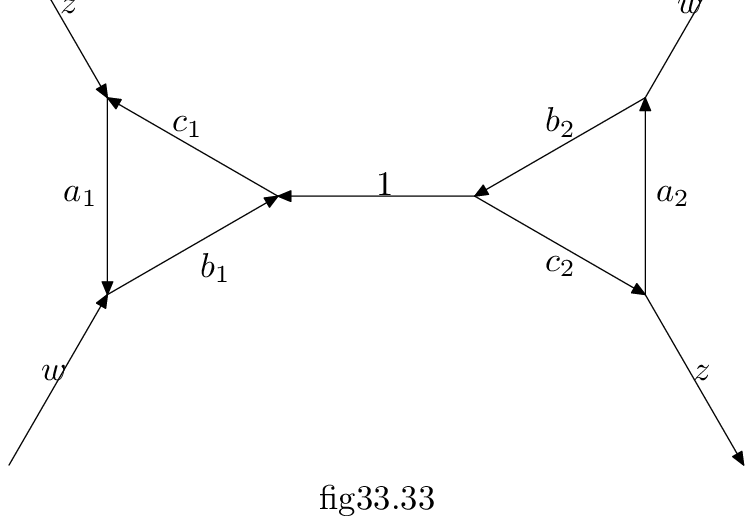}
   \caption{Weighted $1\times1$ Fundamental Domain }
\end{figure}

\begin{align*} P(z,w)&=\det\left(\begin{array}{cccccc} 0&
c_1&  -b_1&   -1& 0& 0\\-c_1&    0&   a_1&    0& -\frac{1}{z}& 0\\
b_1&  -a_1&
0& 0& 0 &-\frac{1}{w}\\     1&    0&    0&    0&   c_2&  -b_2\\
0& z& 0& -c_2& 0&   a_2\\     0&    0&    w&   b_2&  -a_2&
0\end{array}
 \right)\\
&=(z+\frac{1}{z})(ab-c)+(w+\frac{1}{w})(ac-b)+(\frac{z}{w}+\frac{w}{z})(bc-a)+a^2+b^2+1+c^2
\end{align*}
where
\begin{eqnarray}
a=a_1a_2\qquad b=b_1b_2\qquad c=c_1c_2\label{prodweight}
\end{eqnarray}
and we consider the non-degenerate case, which means $ac-b\neq 0$,
$bc-a\neq 0$, $ab-c\neq 0$. We have the following lemma:

\begin{lemma}
For realizable vertex model with period $1\times 1$, let $a,b,c,d$
be the product of dimer weights (\ref{prodweight}) after holographic
reduction, under assumption 3.4, we have the following inequalities:
\begin{align*}
(a+b-c-d)(a+c-b-d)(a+d-b-c)>0\\
a+b+c+d>0.
\end{align*}
\end{lemma}
\begin{proof}
For generic choice of vertex signature, we can assume quotients of
basis entries $a_{ij}=\frac{n_{ij}}{p_{ij}}$ are finite. Apply the
realizability equation (\ref{rlbcd}) to a $1\times 1$ quotient
graph, we have 8 equations for each black vertex, and 8 equations
for each white vertex. At each black vertex, 4 equations has 0 on
the right side, and 4 equations with $a_1,b_1,c_1,d_1$ on the right
side. We divide the 8 equations into 4 groups, each of which has 1
equation with 0 on the right, 1 equation with a non-vanishing edge
weight on the right. We take the difference of the two equations in
each group, and we get 4 new equations with a non-vanishing edge
weight on the right. We perform the same procedure for each white
vertex. Then we multiply and add those equations correspondingly, we
obtain
\begin{align*}
a+d-b-c&=(-a_{03}+a_{13})(y_3x_4+x_8y_7+x_6y_5+y_1x_2)\\
a+c-b-d&=(a_{02}-a_{12})(y_1x_3+y_2x_4+y_6x_8+y_5x_7)\\
a+b-c-d&=(a_{01}-a_{11})(y_1x_5+y_4x_8+y_2x_6+y_3x_7)
\end{align*}
Moreover, since holographic reduction leaves the partition function
invariant,
\begin{align*}
a+b+c+d=x_1y_1+x_2y_2+x_3y_3+x_4y_4+x_5y_5+x_6y_6+x_7y_7+x_8y_8
\end{align*}
where the left side is the partition function of dimers of the
$1\times 1$ quotient graph, and the right side is the partition
function of the vertex model, see Figure 10. According to lemma 5.1,
$a_{0k}$ and $a_{1k}$ are two real numbers with opposite sign.
Without loss of generality, we can assume $a_{02}$ and $a_{03}$ are
both positive. At the black vertices, the $\{000\}$ entry of
matchgate signature corresponds to partition functions of perfect
matchings of the generator with all output vertices kept. Since
there are odd number of vertices, no perfect matching exists,
therefore
\begin{align*}
a_{01}=-\frac{a_{02}a_{03}y_5+a_{02}y_6+a_{03}y_7+y_8}{a_{02}a_{03}y_1+a_{02}y_2+a_{03}y_3+y_4}<0
\end{align*}
we have
\begin{align*}
-(a_{01}-a_{11})(a_{02}-a_{12})(a_{03}-a_{13})>0
\end{align*}
The lemma follows from the assumption that entries of relation
signatures are positive.
\end{proof}

The expression of the characteristic polynomial shows that $P(z,w)$
is a smooth function on $\mathbb{T}^2=\{(z,w)||z|=1,|w|=1\}$. We are
interested in the intersection of spectral curve $P(z,w)=0$ with
unit torus $\mathbb{T}^2$, because it has implications on the
convergence rate of correlations. Theorem 4.4 describes the behavior
of $P(z,w)$ on $\mathbb{T}^2$. Before stating the theorem, we
mention an elementary lemma:

\begin{lemma} Let $f(\phi)=A\sin\phi+B\cos\phi+C$, where $A,B,C$ are real.
If $C\geq 0$ and $A^2+B^2-C^2\leq 0$, then $f(\phi)\geq 0$ for any
$\phi\in\mathbb{R}$.
\end{lemma}

\begin{theorem}
Under assumption 3.4, either the spectral curve after holographic
reduction is disjoint from $\mathbb{T}^2$, or intersects
$\mathbb{T}^2$ at a single real node, that is, for some
$(z_0,w_0)=(\pm1,\pm1)$,
\[P(z,w)=\alpha(z-z_0)^2+\beta(z-z_0)(w-w_0)+\gamma(w-w_0)^2+...,\]
where $\beta^2-4\alpha\gamma\leq 0$.
\end{theorem}
\begin{proof}
Let
\begin{align*}
Q(\theta,\phi)&=P(e^{i\theta},w^{i\phi})\\
&=[2(ac-b)+2(bc-a)\cos\theta]\cos\phi+2(bc-a)\sin\theta\sin\phi\\
&+2(ab-c)\cos\theta+a^2+b^2+c^2+1.
\end{align*}

Consider $Q(\theta,\phi)$ as a trigonometric polynomial with respect
to $\phi$,
\begin{align*}
Q(\theta,\phi)=A\sin\phi+B\cos\phi+C
\end{align*}
where
\begin{align*}
A&=2(ac-b)+2(bc-a)\cos\theta\\
B&=2(bc-a)\sin\theta\\
C&=a^2+b^2+c^2+1+2(ab-c)\cos\theta\geq 0.
\end{align*}
Define
\begin{align*}
g(\cos\theta)&=C^2-A^2-B^2\\
&=4(ab-c)^2\cos^2\theta+4(ab+c)(a^2+b^2-c^2-1)\cos\theta\\
&-4(bc-a)^2-4(ac-b)^2+(a^2+b^2+c^2+1)^2,
\end{align*}
then $g$ is a quadratic polynomial with respect to $\cos\theta$ with
discriminant
\begin{eqnarray*}
\Delta=64abc(a+b-c-1)(a+b+c+1)(a+1-b-c)(a+c-b-1).
\end{eqnarray*}
The minimal value of $g(t)$ is attained at
\[t_0=-\frac{(ab+c)(a^2+b^2-c^2-1)}{2(ab-c)^2},\]
where \[g(t_0)=-\frac{\Delta}{16(ab-c)^2}\]

If $abc>0$,
\begin{align*}
g(\cos\theta)
=[2(ab+c)\cos\theta+a^2+b^2-c^2-1]^2+16abc\sin^2\theta\geq0
\end{align*}
therefore $Q(\theta,\phi)\geq 0$, and $Q(\theta,\phi)=0$ only if
$\sin\theta=\sin\phi=0$.

If $abc<0$, lemma 4.2 implies that $\Delta<0$, then
$g(\cos\theta)>0$, therefore $Q(\theta,\phi)>0$.

If $abc=0$, we have $\left|\frac{ab+c}{ab-c}\right|=1$. Moreover,
\[(a^2+b^2-c^2-1)^2-4(ab-c)^2=(a+1+b+c)(a-1-b+c)(a+1-b-c)(a-1+b-c)>0,\]
which implies
\[|t_0|=\left|\frac{ab+c}{ab-c}\right|\left|\frac{a^2+b^2-c^2-1}{2(ab-c)}\right|>1\]
then $g(\cos\theta)>0$, and $Q(\theta,\phi)>0$.

Therefore the only possible intersection of the spectral curve with
$\mathbb{T}^2$ is real point. It is trivial to check that
$\left.\frac{\partial P}{\partial z}\right|_{(\pm1,\pm1)}=0$,
$\left.\frac{\partial P}{\partial w}\right|_{(\pm1,\pm1)}=0$, and
$Hessian[P(z,w)]$ is positive definite wherever a real intersection
of $P(z,w)=0$ and $\mathbb{T}^2$ exists, then the theorem follows.
\end{proof}

\section{Local Statistics}
\subsection{Configuration at Single Vertex}
We are interested in the probability of a specified configuration at
a single vertex. Consider a realizable vertex model on a planar
finite hexagonal lattice, whose partition function is $S$. Assume
both the vertex with specified signature and its neighbors lie in
the interior of the graph. We work out an explicit example here;
other cases are very similar. If only the configuration $\{011\}$ is
allowed at a black vertex $v$, we get a new vertex model, assume
partition function is $S_{001}$. Then the probability that the local
configuration $\{001\}$ appear at $v$ is
\[Pr(\{001\},v)=\frac{S_{001}}{S}.\]
Hence the local statistics problem reduces to the problem of how to
compute $S_{001}$ efficiently. We will show that $S_{001}$ can also
be computed by the dimer technique with the help of the holographic
algorithm.

 Let us denote the adjacent
vertices of $v$ by $v_a,v_b,v_c$, according to the direction of the
connecting edges, as illustrated in Figure 11.
\begin{figure}[htbp]
  \centering
\includegraphics*[266,756][350,818]{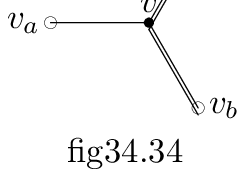}
   \caption{Configuration {011} at a Black Vertex }
\end{figure}

Let $r_v,\ r_a,\ r_b,\ r_c$ denote the signatures of relations at
vertices $v,\ v_a,\ v_b,\ v_c$, then we have
\begin{align*}
\left(\begin{array}{c}{r_v}_{\{000\}}\\{r_v}_{\{001\}}\\{r_v}_{\{010\}}\\{r_v}_{\{011\}}\\{r_v}_{\{100\}}\\{r_v}_{\{101\}}\\{r_v}_{\{110\}}\\{r_v}_{\{111\}}\end{array}\right)=\left(\begin{array}{c}0\\0\\0\\y_4\\0\\0\\0\\0\end{array}\right)\qquad
\left(\begin{array}{c}{r_a}_{\{000\}}\\{r_a}_{\{001\}}\\{r_a}_{\{010\}}\\{r_a}_{\{011\}}\\{r_a}_{\{100\}}\\{r_a}_{\{101\}}\\{r_a}_{\{110\}}\\{r_a}_{\{111\}}\end{array}\right)=\left(\begin{array}{c}x_1^a\\x_2^a\\
x_3^a\\x_4^a\\\star\\\star\\\star\\\star \end{array}\right)\\
\left(\begin{array}{c}{r_b}_{\{000\}}\\{r_b}_{\{001\}}\\{r_b}_{\{010\}}\\{r_b}_{\{011\}}\\{r_b}_{\{100\}}\\{r_b}_{\{101\}}\\{r_b}_{\{110\}}\\{r_b}_{\{111\}}\end{array}\right)=\left(\begin{array}{c}\star\\\star
\\x_3^b\\x_4^b\\\star \\\star \\x_7^b\\x_8^b\end{array}\right)\qquad
\left(\begin{array}{c}{r_c}_{\{000\}}\\{r_c}_{\{001\}}\\{r_c}_{\{010\}}\\{r_c}_{\{011\}}\\{r_c}_{\{100\}}\\{r_c}_{\{101\}}\\{r_c}_{\{110\}}\\{r_c}_{\{111\}}\end{array}\right)=\left(\begin{array}{c}
\star\\x_2^c\\\star
\\x_4^c\\\star \\x_6^c\\\star\\x_8^c \end{array}\right)\\
\end{align*}
Here $\star$ means the entry at the position can be arbitrary. This
is because we only allow the configuration $\{011\}$ at $v$,
therefore the only configurations which actually affect the
partition function of the vertex model will be those who do not
occupy the a-edge of $v_a$, and occupy both the b-edge for $v_b$ and
the c-edge for $v_c$. We will split each of $r_a,r_b,r_c$ into 2
parts, namely, $r_l=r_l(0)+r_l(1),\ l=a,b,c$, below. By definition,
the partition function of the vertex model with signature
$r,r_a,r_b,r_c$ is a sum of 8 terms, each of which is the partition
function of a vertex model with signature $r_v,r_a(i),r_b(j),r_c(k)\
i,j,k\in\{0,1\}$. That is,
\begin{eqnarray}
S_{001}=S_{\{r_v,r_a,r_b,r_c\}}}=\sum_{i,j,k\in\{0,1\}}S_{\{r_v,r_a(i),r_b(j),r_c(k)\}\label{splitvertexmodel}
\end{eqnarray}
For each $S_{\{r_v,r_a(i),r_b(j),r_c(k)\}}$, we give new bases on
edges adjacent to $v$, such that $S_{\{r_v,r_a(i),r_b(j),r_c(k)\}}$
become the partition function of certain local dimer configurations.
Namely
\begin{eqnarray}
S_{\{r_v,r_a(i),r_b(j),r_c(k)\}}\buildrel base\ change \over
=Z_{\{m_v(i,j,k),m_a(i),m_b(j),m_c(k)\}}\label{basechange}
\end{eqnarray}
where $Z_{\{m_v(i,j,k),m_a(i),m_b(j),m_c(k)\}}$ is the partition
function of the dimer model on a matchgrid with signature
$m_v(i,j,k),m_a(i),m_b(j),m_c(k)$. Assume the original matchgates
$u_a,\ u_b,\ u_c$ have weights
\begin{align*}
m_l=\left(\begin{array}{cccccccc}0&a_2^l&b_2^l&0&c_2^l&0&0&d_2^l\end{array}\right)^t,\
\ \ l=a,b,c
\end{align*}
see Figure 12.
\begin{figure}[htbp]
  \centering
\includegraphics*[238,673][379,820]{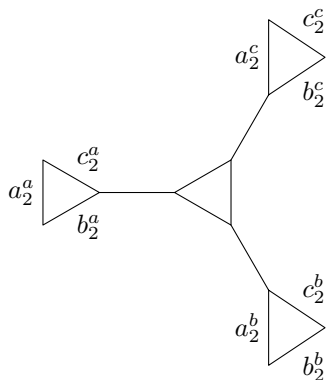}
   \caption{Weighted Matchgate }
\end{figure}
We require
\begin{eqnarray*}
m_a(1)=\left(\begin{array}{cccccccc}0&0&0&0&c_2^a&0&0&d_2^a\end{array}\right)^t\\
m_a(0)=\left(\begin{array}{cccccccc}0&a_2^a&b_2^a&0&0&0&0&0\end{array}\right)^t\\
m_b(1)=\left(\begin{array}{cccccccc}0&0&b_2^b&0&0&0&0&d_2^b\end{array}\right)^t\\
m_b(0)=\left(\begin{array}{cccccccc}0&a_2^b&0&0&c_2^b&0&0&0\end{array}\right)^t\\
m_c(1)=\left(\begin{array}{cccccccc}0&a_2^c&0&0&0&0&0&d_2^c\end{array}\right)^t\\
m_c(0)=\left(\begin{array}{cccccccc}0&0&b_2^c&0&c_2^c&0&0&0\end{array}\right)^t.\label{dimersignature}
\end{eqnarray*}

Notice that for any $l\in\{a,b,c\}$, $m_l(0)$ corresponds to
configurations with an unoccupied $l$-edge; and $m_l(1)$ corresponds
to configurations with an occupied $l$-edge. We are trying to
determine the new base change matrices on incident edges of $v$, for
each $r_v,r_a(i),r_b(j),r_c(k)$, such that
(\ref{splitvertexmodel}),(\ref{basechange}), (\ref{dimersignature})
are satisfied simultaneously.

Assume on the adjacent edges of $v_a,\ v_b,\ v_c$ we have original
base change matrix
\begin{align*}
T_k^j=\left(\begin{array}{cc}n_{0k}^j&n_{1k}^j\\p_{0k}^j&p_{1k}^j\end{array}\right)\
\ \ k=1,2,3;\ j=a,b,c
\end{align*}

On an a-type edge adjacent to $v$, we consider 2 possible base
change matrices
\begin{align*}
S_1^0=\left(\begin{array}{cc}m_{01}^0&0\\q_{01}^0&q_{11}^0\end{array}\right)\qquad
S_1^1=\left(\begin{array}{cc}0&m_{11}^1\\q_{01}^1&q_{11}^1\end{array}\right)
\end{align*}
On a b-type edge adjacent to $v$, we consider 2 possible base change
matrices
\begin{align*}
S_2^0=\left(\begin{array}{cc}m_{02}^0&m_{12}^0\\q_{02}^0&0\end{array}\right)\qquad
S_2^1=\left(\begin{array}{cc}m_{02}^1&m_{12}^1\\0&q_{12}^1\end{array}\right)
\end{align*}
On a c-type edge adjacent to $v$, we consider 2 possible base change
matrices
\begin{align*}
S_3^0=\left(\begin{array}{cc}m_{03}^0&m_{13}^0\\q_{03}^0&0\end{array}\right)\qquad
S_3^1=\left(\begin{array}{cc}m_{03}^1&m_{13}^1\\0&q_{13}^1\end{array}\right)
\end{align*}

For signatures $r_v,r_a(i),r_b(j),r_c(k)$, we choose basis
$S_1^i,S_2^j,S_3^k$ on edges $a,b,c$. By the realizability condition
at vertex $v$, we have
\begin{align*}
m_v(i,j,k)=(S_1^i\otimes S_2^j\otimes S_3^k)^t\cdot
r_v=e_{\{ijk\}}w_{ijk},
\end{align*}
where
\[w_{ijk}=m_{i1}^iq_{j2}^jq_{k3}^ky_4,\]
and $e_{\{ijk\}}$ is the 8 dimensional vector with entry 1 at the
position labeled by the binary sequence $\{ijk\}$, and 0 elsewhere.
Obviously the choice of new base change matrices, namely, the
position of zeros in the new base change matrices, results in the
single configuration $\{ijk\}$ at the matchgate corresponding to
$v$. Now the problem is to determine the entries of the base change
matrices satisfying the equations.

According to realizability conditions at vertices $v_a,v_b,v_c$,
\begin{eqnarray}
r_a&=&\Sigma_{i=1}^2 r_a(i)=\Sigma_{i=1}^2 S_1^i\otimes T_2^a\otimes
T_3^a\cdot m_a(i)\\
r_b&=&\Sigma_{j=1}^2 r_a(j)=\Sigma_{i=1}^2 T_1^b\otimes S_2^j\otimes
T_3^b\cdot m_b(j)\\
r_c&=&\Sigma_{k=1}^2 r_c(k)=\Sigma_{k=1}^2 T_1^c\otimes T_2^c\otimes
S_3^k\cdot m_c(k).\label{splitrealizability}
\end{eqnarray}

For the original basis, we have the realizability condition as
follows:
\begin{eqnarray}
r_l=T_1^l\otimes T_2^l\otimes T_3^l\cdot
m_l.\label{originalrealizability}
\end{eqnarray}
Substituting $r_a,r_b,r_c$ by (\ref{originalrealizability}) in
(\ref{splitrealizability}) we have the following system of linear
equations with respect to entries of $S$
\begin{align*}
T_2^a\otimes
T_3^a\left(\begin{array}{c}m_{11}^1c_2^a\\m_{01}^2a_2^a\\m_{01}^2b_2^a\\m_{11}^1d_2^a\end{array}\right)=\left(\begin{array}{c}x_1^a\\x_2^a\\x_3^a\\x_4^a\end{array}\right)=T_2^a\otimes
T_3^a\left(\begin{array}{c}n_{11}^ac_2^a\\n_{01}^aa_2^a\\n_{01}^ab_2^a\\n_{11}^ad_2^a\end{array}\right)\\
T_1^b\otimes
T_3^b\left(\begin{array}{c}q_{12}^1b_2^b\\q_{02}^2a_2^b\\q_{02}^2c_2^b\\q_{12}^1d_2^b\end{array}\right)=\left(\begin{array}{c}x_3^b\\x_4^b\\x_7^b\\x_8^b\end{array}\right)=T_1^b\otimes
T_3^b\left(\begin{array}{c}p_{12}^bb_2^b\\p_{02}^ba_2^b\\p_{02}^bc_2^b\\p_{12}^bd_2^b\end{array}\right)\\
T_1^c\otimes
T_2^c\left(\begin{array}{c}q_{13}^1a_2^c\\q_{03}^2b_2^c\\q_{03}^2c_2^b\\q_{13}^1d_2^b\end{array}\right)=\left(\begin{array}{c}x_2^c\\x_4^c\\x_6^c\\x_8^c\end{array}\right)=T_1^c\otimes
T_2^c\left(\begin{array}{c}p_{13}^ca_2^c\\p_{03}^cb_2^c\\p_{03}^cc_2^b\\p_{13}^cd_2^b\end{array}\right)
\end{align*}
Under the assumption that original base change matrices are
invertible, we have
\begin{align*}
m_{11}^1=n_{11}^a,\qquad m_{01}^2=n_{01}^a\\
q_{12}^1=p_{12}^b,\qquad q_{02}^2=p_{02}^b\\
q_{13}^1=p_{13}^c,\qquad q_{03}^2=p_{03}^c
\end{align*}
Therefore, we only need to choose the nonzero entries of the new
base change matrices to be the same as the original ones.

For the other seven configurations at vertex $v$, we can use the
same method to achieve a similar result. By splitting each one of
$r_a$, $r_b$, $r_c$ into 2 parts, we express the partition function
of the satisfying assignments as the sum of 8 terms, each one of
which is the partition function of the relations with signature
$r_v,\ r_a(i),\ r_b(j), r_c(k)$, $i,j,k\in\{0,1\}$. We apply the
holographic reduction to each part separately, and we derive that
the partition function of the relations with signature $r_v,\
r_a(i),\ r_b(j), r_c(k)$, $i,j,k\in\{0,1\}$ is equal to the
partition function of the dimer configurations on the corresponding
matchgrids with signature $m_v(i,j,k),m_a(i),m_b(j),m_c(k)$, which
corresponds to the local configuration $\{ijk\}$. This may not be a
dimer configuration; to make it a dimer configuration, let us divide
those 8 terms into two groups: one consists of all $\{i,j,k\}$
satisfying $i+j+k\equiv 0(\mod\ 2)$, namely
$\{000\},\{011\},\{101\},\{110\}$; the other consists of all
$\{i,j,k\}$ satisfying $i+j+k\equiv 1(\mod\ 2)$. The sum of
partition functions in the first group is equal to the partition
function of dimer configurations in a matchgrid with weights
illustrated in the left graph of Figure 11, where
$t=\frac{w_{000}}{w_{011}+w_{101}+w_{110}}$; the sum of partition
functions corresponding to $\{001\},\{010\},\{100\},\{111\}$ is
equal to the partition function of dimer configurations in a
matchgrid with weights as illustrated in the right graph of Figure
11, up to a multiplication constant $w_{111}$.

\begin{figure}[htbp]
  \centering
\includegraphics*[232,672][379,821]{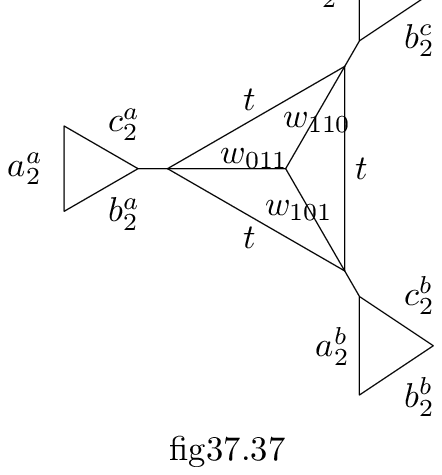}\qquad\includegraphics*[232,672][379,821]{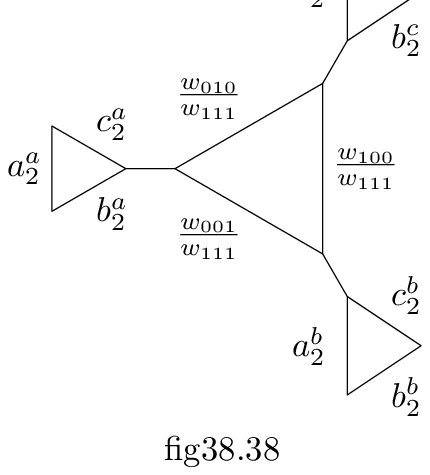}
   \caption{Even and Odd Matchgates}
\end{figure}

However, in the left graph we change the parity of the total number
of vertices, as a result, there is no dimer. Therefore we have the
following theorem

\begin{theorem}
Assume we have a vertex model $\mathcal{H}_0$ on a finite hexagonal
lattice, which is holographic equivalent to the dimer model on a
Fisher graph $\mathcal{F}_0$. If at a fixed interior vertex $v$ of
$\mathcal{H}_0$, only one local configuration is allowed, then the
partition function of the new vertex model $\mathcal{H}_1$ can be
computed by the partition function of the dimer model on a Fisher
graph $\mathcal{F}_1$, multiplied by a non-zero constant.
$\mathcal{F}_0$ and $\mathcal{F}_1$ have the same edge weights
except for the matchgate corresponding to $v$.
\end{theorem}

\subsection{Configuration at Finitely Many Vertices}
We are interested in the probability that one specified
configuration is allowed at each of the given finitely many vertices
for a realizable vertex model on a hexagonal lattice.  For
simplicity, assume all of them are black. Let $V_0^b$ be the set of
all black vertices with specified configuration. To compute the
probability, we first give a criterion to construct the new base
change matrices on incident edges of $V_0^b$ according to whether
the edge is present in the dimer configuration and in the vertex
model configuration. Assume
$\left(\begin{array}{cc}n_0&n_1\\p_0&p_1\end{array}\right)$ is the
original base change matrix on the edge,
\\
\newcommand{\ZZ}[2]{\rule[#1]{0pt}{#2}}
\begin{center}
\begin{tabular}{|c|c|c|}\hline
\multicolumn{1}{|c|}{\ZZ{-8pt}{25pt}\textbf{Base Change Matrices}}
&\multicolumn{1}{|c|} {\ZZ{-8pt}{25pt}\textbf{Presence in Vertex
Configuration}}
&\multicolumn{1}{|c|}{\ZZ{-8pt}{25pt}\textbf{Presence in
Dimer}}\\\hline
\multicolumn{1}{|c|}{\small$\left(\begin{array}{cc}0&
n_1\\p_0&p_1\end{array}\right)$}&\multicolumn{1}{|c|}{No}&\multicolumn{1}{|c|}{Yes}\\\hline
\multicolumn{1}{|c|}{\small$\left(\begin{array}{cc}n_0&
0\\p_0&p_1\end{array}\right)$}&\multicolumn{1}{|c|}{No}&\multicolumn{1}{|c|}{No}\\\hline
\multicolumn{1}{|c|}{\small$\left(\begin{array}{cc}n_0&
n_1\\0&p_1\end{array}\right)$}&\multicolumn{1}{|c|}{Yes}&\multicolumn{1}{|c|}{Yes}\\\hline
\multicolumn{1}{|c|}{\small$\left(\begin{array}{cc}n_0&
n_1\\p_0&0\end{array}\right)$}&\multicolumn{1}{|c|}{Yes}&\multicolumn{1}{|c|}{No}\\\hline
\end{tabular}
\end{center}
whether the edges incident to vertices in $V_0^b$ are present in the
configuration of the vertex model is known, given all the specified
configurations at $V_0^b$. As before, we are going to split the
signature of each adjacent white vertex into several parts, and we
apply the holographic reduction to each parts separately. Our
expectation is that after the reduction process, each part will be
equivalent to the partition function of a single local dimer
configuration. The new base change matrix on each edge is chosen
according to whether we want the edge to be present in the dimer
configuration or not after the reduction. As before, all the
non-vanishing entries of the new bases change matrices are equal to
the original ones; we will see how it works below.

Consider an arbitrary white vertex $w\in \Gamma(V_0^b)$, the
neighbors of $V_0^b$. Assume $r_w$ is the vertex signature at $w$,
and $m_w$ is the original dimer signature at the corresponding
matchgate. Assume
\begin{align*}
m_w=\left(\begin{array}{cccccccc}0&a_w&b_w&0&c_w&0&0&d_w\end{array}\right)^t
\end{align*}
Let $D(w)$ denote the number of adjacent vertices of $w$ with
specified configuration. We classify $\Gamma(V_0^b)$ according to
$D$.

If $D(w)=1$, we split $r_w=r_w(1)+r_w(2)$. Without loss of
generality, assume the left-digit edge(a-type edge, horizontal edge)
connects $w$ with $b\in V_0^b$, and the edge $wb$ is present in the
configuration of the vertex model. Assume
\begin{align*}
T_1=\left(\begin{array}{cc}n_{01}&n_{11}\\p_{01}&p_{11}\end{array}\right)\qquad
T_2=\left(\begin{array}{cc}n_{02}&n_{12}\\p_{02}&p_{12}\end{array}\right)\qquad
T_3=\left(\begin{array}{cc}n_{03}&n_{13}\\p_{03}&p_{13}\end{array}\right)
\end{align*}
are original base change matrices on the edges adjacent to $w$.
According to the table, the presence of $wb$ in the vertex model
configuration implies two possible choices of the new basis on the
horizontal edge, namely,
\begin{align*}
S_1^1=\left(\begin{array}{cc}n_{01}&n_{11}\\p_{01}&0\end{array}\right)\qquad
S_1^2=\left(\begin{array}{cc}n_{01}&n_{11}\\0&p_{11}\end{array}\right)
\end{align*}
while on the two other incident edges of $w$ we have the original
bases $T_2$, $T_3$. Assume
\begin{align*}
m_w(1)=\left(\begin{array}{cccccccc}0&a_w&b_w&0&0&0&0&0\end{array}\right)\\
m_w(2)=\left(\begin{array}{cccccccc}0&0&0&0&c_w&0&0&d_w\end{array}\right)
\end{align*}
Notice that the nonzero entries $m_w(1)$ have indices $\{001\}$ and
$\{010\}$, which corresponds to configurations without the a-edge
present; the non-vanishing entries of $m_w(2)$ have indices
$\{100\}$ and $\{111\}$, which corresponds to configurations with
the a-edge present. Choose
\begin{align*}
r_w(1)=S_1^1\otimes T_2\otimes T_3\cdot m_w(1)\\
r_w(2)=S_1^2\otimes T_2\otimes T_3\cdot m_w(2)
\end{align*}
Then we can check
\begin{align*}
r_w(1)+r_w(2)=\otimes_{i=1}^{3}T_i\cdot m_w=r_w
\end{align*}

If $D(w)=2$, we split $r_w=\sum_{i=1}^{4}r_w(i)$. Without loss of
generality, assume the middle-digit(b-type edge) and
right-digit(c-type edge) edges connects $w$ to $b_2,b_3\in V_0^b$,
$wb_2$ is present while $wb_3$ is not present in the configuration
of the vertex model. According to our criteria listed in the table,
on $wb_2$ we have two different base change matrices
\begin{align*}
S_2^1=\left(\begin{array}{cc}n_{02}&n_{12}\\p_{02}&0\end{array}\right)\qquad
S_2^2=\left(\begin{array}{cc}n_{02}&n_{12}\\0&p_{12}\end{array}\right)
\end{align*}
On $wb_3$ we have two different base change matrices
\begin{align*}
S_3^1=\left(\begin{array}{cc}n_{03}&0\\p_{03}&p_{13}\end{array}\right)\qquad
S_3^2=\left(\begin{array}{cc}0&n_{13}\\p_{03}&p_{13}\end{array}\right)
\end{align*}
while on the left-digit edge, we keep original basis $T_1$, the
original basis on the middle and right digit edge are $T_2$, $T_3$,
as in the $D(w)=1$ case. Assume
\begin{align*}
m_w(1)=\left(\begin{array}{cccccccc}0&a_w&0&0&0&0&0&0\end{array}\right)\\
m_w(2)=\left(\begin{array}{cccccccc}0&0&b_w&0&0&0&0&0\end{array}\right)\\
m_w(3)=\left(\begin{array}{cccccccc}0&0&0&0&c_w&0&0&0\end{array}\right)\\
m_w(4)=\left(\begin{array}{cccccccc}0&0&0&0&0&0&0&d_w\end{array}\right)
\end{align*}
Obviously, $m_w(4)$ corresponds to both edge being present, $m_w(3)$
corresponds to neither being present, $m_w(2)$ corresponds to b-edge
being present and c-edge not, and $m_w(1)$ corresponds to b-edge not
present and c-edge present. Again according to the table, choose
\begin{align*}
r_w(1)=T_1\otimes S_2^1\otimes S_3^2\cdot m_w(1)\\
r_w(2)=T_1\otimes S_2^2\otimes S_3^1\cdot m_w(2)\\
r_w(3)=T_1\otimes S_2^1\otimes S_3^1\cdot m_w(3)\\
r_w(4)=T_1\otimes S_2^2\otimes S_3^2\cdot m_w(4)
\end{align*}
we can check
\begin{align*}
\sum_{i=1}^{4}r_w(i)=\otimes_{1}^{3}T_{i}\cdot m_w=r_w
\end{align*}

If $D(w)=3$, we split $r_w=\sum_{i=1}^{4}r_w(i)$, and assume
$m_{w}(i),i=1,...,4$ as in $D(w)=2$. Without loss of generality,
assume $S_1^1,S_1^2$ as in $D(w)=1$,
$S_2^{1},S_2^{2},S_3^{1},S_3^{2}$ as in $D(w)=2$; that is, we have
vertex-model signature $\{110\}$ at $w$. Choose
\begin{align*}
r_w(1)=S_1^{1}\otimes S_2^1\otimes S_3^2\cdot m_w(1)\\
r_w(2)=S_1^{1}\otimes S_2^2\otimes S_3^1\cdot m_w(2)\\
r_w(3)=S_1^{2}\otimes S_2^1\otimes S_3^1\cdot m_w(3)\\
r_w(4)=S_1^{2}\otimes S_2^2\otimes S_3^2\cdot m_w(4)
\end{align*}
again we can check $r_w=\sum_{i=1}^{4}r_w(i)$ as in $D(w)=2$.

For all the other local configurations, the same technique works,
and we will have a similar result. Assume $V_0^b=\{v_1,...,v_p\},
\Gamma(V_0^b)=\{w_1,...,w_k\}$, then
\begin{align*}
S(\tilde{r}_{v_1},...,\tilde{r}_{v_p})
&=&[\otimes_{j=1}^{k}r_{w_j}\otimes
\otimes_{w\notin\Gamma(V_0^b)}r_w]\cdot[\otimes_{q=1}^{p}\tilde{r}_{v_q}\otimes \otimes_{b\notin V_0^b}r_b]\\
&=&\sum_{i_1,...,i_k}[\otimes_{j=1}^{k}r_{w_j}(i_j)\otimes
\otimes_{w\notin\Gamma(V_0^b)}r_w]\cdot[\otimes_{q=1}^{p}\tilde{r}_{v_q}\otimes
\otimes_{b\notin V_0^b}r_b]
\end{align*}
where $\tilde{r}_{v_q}$ is the specified configuration at the vertex
$v_q$. The first equality follows from the definition of the
partition function of satisfying assignments. When we compute the
tensor product of relation signatures of all black (white) vertices,
we get a vector of dimension $2^{|E(G)|}$, where $|E(G)|$ is the
number of edges in the planar finite graph $G$. This vector can be
indexed by binary sequences of length $|E(G)|$. Each binary sequence
corresponds to a configuration, and the entry there is the product
of weights of the local configurations obtained from the
configuration restricted to each black (white) vertices. Obviously
the inner product of the vector at black vertices and the vector at
white vertices are exactly the partition function of the satisfying
assignments. The second equality follows from the multi-linearity of
the tensor product. For $1\leq j\leq k$, if $D(w_j)=1$,
$i_j\in\{1,2\}$; if $D(w_j)=2,3$, $i_j\in\{1,2,3,4\}$. For each
summand, we choose a basis on incident edges of vertices in $V_0^b$
according to whether the edge is present in the dimer configuration
and relation configuration, and keep the original bases on all the
other edges. as described above. This is realizable because for each
$b\in V_0^b$, whether its incident edges are occupied by dimers are
completely determined in each part of the sum. In other words, each
term on the right side of the second equality corresponds to an
configuration on all the edges incident to vertices in $V_0^b$,
which is a local dimer configuration at each odd matchgate
corresponding to white vertices $w$'s with $D(w)=3$, To make them be
dimer configurations at each black matchgate, we divide those
configurations into groups according to the following criterion: two
configurations are in the same group if and only if the parity of
the number of occupied incident edges at each black vertex in
$V_0^b$ is the same. If zero or two incident edges are occupied, we
construct an even matchgate with modified weights at the black
vertex; if one or three incident edges are occupied, we construct an
odd matchgate with modified weights. Notice that the modified
weights depend only on the local configuration on the incident edges
of the vertex. Since $|V_0^b|=p$, we have $2^p$ different
constructions in total, but only $2^{p-1}$ of them admit dimer
cover, each of which has even number of even matchgates. Therefore
we have

\begin{theorem}
Assume we have a realizable vertex model $\mathcal{H}_0$ on a finite
hexagonal lattice, which is holographic equivalent to the dimer
model on  a Fisher graph $\mathcal{F}_0$. Let $V_0^b$ be a subset of
black vertices. If for each $v\in V_0^b$, we only allow one local
configuration, this way we obtain a new vertex model
$\mathcal{H}_1$. The partition function of $\mathcal{H}_1$ is equal
to the sum of partition functions of $2^{p-1}$ dimer models on
matchgrids $\mathcal{F}_1,...,\mathcal{F}_{2^{p-1}}$, each of which
has an even number of even matchgates. $\mathcal{F}_i(1\leq i\leq
2^{p-1})$ are the same $\mathcal{F}_0$, except for the matchgates
corresponding to vertices in $V_0^b$.
\end{theorem}

\subsection{Ising Model and Vertex Model}

Consider the Ising model on a finite Kagome lattice, embedded into a
torus, as illustrated in the following figure.

\begin{figure}[htbp]
  \centering
\includegraphics*[171,659][445,815]{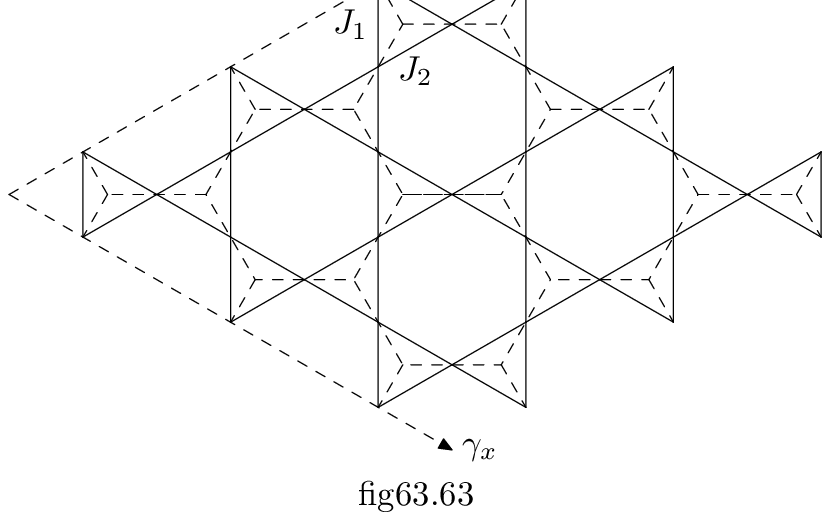}
   \caption{Kagome Lattice and Honeycomb Lattice}
\end{figure}

The associated honeycomb lattice is illustrated in the Figure by the
dashed line. Each vertex of the Kagome lattice corresponds to an
edge of the honeycomb lattice; hence each Ising spin configuration
on the Kagome lattice corresponds to an edge subset of the honeycomb
lattice. If the spin is ''$+$'', then the corresponding edge is
included in the subset; otherwise the edge is not included. Assume
the bonds of the Kagome lattice have interactions $J_1,J_2,J_3$, as
illustrated in Figure 14. We define a vertex model on the honeycomb
lattice with signature at all vertices as follows.
\begin{align*}
\left(\begin{array}{c}r_{v,000}\\r_{v,001}\\r_{v,010}\\r_{v,011}\\r_{v,100}\\r_{v,101}\\r_{v,110}\\r_{v,111}\end{array}\right)=\left(
\begin{array}{c}e^{2(J_1+J_2+J_3)}\\e^{2J_3}\\e^{2J_2}\\e^{2J_1}\\e^{2J_1}\\e^{2J_2}\\e^{2J_3}\\e^{2(J_1+J_2+J_3)}
\end{array}\right)
\end{align*}
This way the probability measure of the Ising model on the Kagome
lattice is equivalent to the probability measure of the vertex model
on the honeycomb lattice. It is trivial to check that this vertex
model is orthogonally realizable.

\section{Asymptotic Behavior}

In this section, we prove the main theorems concerning the
asymptotic behavior of realizable vertex models on the periodic
hexagonal lattice, as stated in the introduction. Consider an
infinite periodic graph $G$, with period $1\times 1$, see Figure 9.
Our technique to deal with such a graph is to consider a graph $G_n$
with $n^2$ $1\times 1$ fundamental domains, embed $G_n$ into a
torus, and consider the limit when $n\rightarrow\infty$, see Page 5.
Our first theorem is about the free energy of the infinite periodic
hexagonal lattice.
\\
\\
\noindent \textbf{Proof of Theorem 1.1} Assume $M_n$ is the
corresponding matchgrid with respect to $G_n$, and $P_n(z,w)$ is the
characteristic polynomial. Obviously $M_n$ is also a quotient graph
of a periodic infinite graph modulo a subgraph of $\mathbb{Z}^2$
generated by $(n,0)$ and $(0,n)$. The corresponding Kasteleyn
matrices here are defined given the orientation of Figure 9. For
even $n$, the crossing orientation can be obtained from the
orientation of Figure 9 by reversing all the $z$-edges and
$w$-edges. By Theorem 3.6, $M_n$ is a Fisher graph, and the
partition of the vertex model can be expressed as follows according
to the principle of holographic reduction:
\begin{align*}
S(G_n)=Z(M_n)=\frac{1}{2}|-\mathrm{Pf} K_n(1,1)+\mathrm{Pf}
K_n(1,-1)+\mathrm{Pf} K_n(-1,1)+\mathrm{Pf} K_n(-1,-1)|
\end{align*}
Thus
\begin{align*}
\max_{u,v\in\{-1,1\}}|\mathrm{Pf} P_n(u,v)|\leq Z(M_n)\leq
2\max_{u,v\in\{-1,1\}}|\mathrm{Pf} P_n(u,v)|
\end{align*}
On the other hand, according to the formula of enlarging fundamental
domains,
\begin{align*}
\frac{1}{n^2}\log\max_{u,v\in\{-1,1\}}|\mathrm{Pf}
K_n(u,v)|=\max_{u,v\in\{-1,1\}}\frac{1}{2n^2}\sum_{z^n=u}\sum_{w^n=v}\log|P(z,w)|
\end{align*}
By Theorem 4.4, either $P(z,w)$ has no zeros on $\mathbb{T}^2$, or
it has a single real node, in which case any sample point in
$\max_{u,v\in\{-1,1\}}\frac{1}{2n^2}\sum_{z^n=u}\sum_{w^n=v}\log|P(z,w)|$
is at least $\frac{C}{n}$ from the real node, for some constant
$C>0$. Therefore
\begin{align*}
\lim_{n\rightarrow\infty}\frac{1}{n^2}S(G_n)&=&\lim_{n\rightarrow\infty}\max_{u,v\in\{-1,1\}}\frac{1}{2n^2}\sum_{z^n=u}\sum_{w^n=v}\log|P(z,w)|\\
&=&\frac{1}{8\pi^2}\iint_{|z|=1,|w|=1}\log
P(z,w)\frac{dz}{iz}\frac{dw}{iw}
\end{align*}$\Box$

For each $G_n$, a measure $\lambda_n$ is defined as in
\eqref{measure}. Assume at a fixed vertices $v$, we only allow
configuration $c_v$. Let $\mu_n$ be the Boltzmann measure of dimer
configurations on $M_n$. Let $\tilde{M}_{n}$ be the matchgrid
corresponding to the vertex model which only allows $c_v$ at $v$, as
described in Theorem 5.1. Let $m_v$ be the matchgate of
$\tilde{M}_n$, corresponding to $v$, and $d_{j}$ be a local dimer
configuration at $m_v$, $w_{d_{j}}$ be product of weights of
matchgate edges included in $d_{j}$, and $V_{d_{j}}$ be the set of
external vertices of matchgates $m_{v}$ occupied by dimer
configuration $d_{j}$. In our graph, every matchgate has 3 external
vertices. Our second theorem is about the asymptotic behavior of the
measure $\lambda_n$.
\\
\\
\noindent \textbf{Proof of Theorem 1.2} According to Theorem 5.1
\begin{align*}
\lambda_n(c_1,...,c_p)&=&\frac{Z(\tilde{M}_{n})}{Z(M_n)}
=\sum_{d_{j}}\frac{Z(\tilde{M}_{n}(d_{j}))}{Z(M_n)}\\
&=&\sum_{d_{j}}w_{d_{j}}\frac{Z(M_n\setminus V(d_{j}))}{Z(M_n)}\\
\end{align*}
where the sum is over all local dimer configurations $d_j$. Since
the number of local configurations is finite, it suffices to
consider $\lim_{n\rightarrow\infty}\frac{Z(M_n\setminus
V(d_{j}))}{Z(M_n)}$. Note that $\tilde{M}_n$ differs from $M_n$ only
on edge weights of $m_v$, hence $M_n\setminus V(d_{j})$ and
$\tilde{M}_n\setminus V(d_{j})$ are the same.

Given $d_{j}$, let $W_{n,d_{j}}=M_n\setminus V(d_{j})$ be the
subgraph of $M_n$ by removing all vertices occupied by the
configuration $(d_{j})$, as well as their incident edges. Then
\begin{align*}
\frac{Z(W_{n,d_{j}})}{Z(M_n)}=\frac{1}{2Z_n}|-\mathrm{Pf}(K_n^{11}(W_{n,d_{j}}))+\mathrm{Pf}(K_n^{-1,1}(W_{n,d_{j}}))\\
+\mathrm{Pf}(K_n^{1,-1}(W_{n,d_{j}}))+\mathrm{Pf}(K_n^{-1,-1}(W_{n,d_{j}}))|
\end{align*}
First of all, let us assume $P(z,w)$ has no zeros on $\mathbb{T}^2$.
According to the formula of enlarging fundamental domains, for any
$(\theta,\tau)\in\{-1,1\}$, $Pf(K_n^{\theta,\tau})\neq 0$. Then
\begin{align*}
|\mathrm{Pf}(K_n^{\theta,\tau})_{E^c}|=|\mathrm{Pf}(K_n^{\theta,\tau})^{-1}_{E}||\mathrm{Pf}(K_n^{\theta,\tau})|
\end{align*}
In \cite{ckp,kos}, it was proved that given two vertices
$(u,x_1,y_1)$ and $(v,x_2,y_2)$
\begin{align*}
(K_n^{\theta,\tau})^{-1}((u,x_1,y_1),(v,x_2,y_2))=\frac{1}{n^2}\sum_{z^n=\theta}\sum_{w^n=\tau}z^{x_1-x_2}w^{y_1-y_2}\frac{\mathrm{Cof}(K(z,w))_{u,v}}{P(z,w)}
\end{align*}
Since $P(z,w)$ has no zeros on $\mathbb{T}^2$, we have
\begin{align*}
\lim_{n\rightarrow\infty}(K_n^{\theta,\tau})^{-1}((u,x_1,y_1),(v,x_2,y_2))&=&\frac{1}{4\pi^2}\iint_{\mathbb{T}^2}z^{x_1-x_2}w^{y_1-y_2}\frac{\mathrm{Cof}(K(z,w))_{u,v}}{P(z,w)}\frac{dz}{iz}\frac{dw}{iw}\\
&=&K_{\infty}^{-1}((u,x_1,y_1),(v,x_2,y_2))
\end{align*}
As $n\rightarrow\infty$, each entry of $(K_n^{\theta\tau})^{-1}$ is
convergent, so is the Pfaffian of a finite order sub-matrix
$(K_n^{\theta\tau})^{-1}_{V(d_{j})}$, and we have
\begin{align*}
\lim_{n\rightarrow\infty}\frac{Z(M_n\setminus
V(d_{j}))}{Z(M_n)}=|\mathrm{Pf}(K_{\infty}^{-1})_{V(d_{j})}|
\end{align*}

If $P(z,w)$ has a real node on $\mathbb{T}^2$, without loss of
generality, we can assume the real node is $(1,1)$. It was proved in
\cite{bt} that if $P(z,w)=0$ has a node at $(1,1)$, then for any
fixed finite subset E
\begin{align*}
\lim_{n\rightarrow\infty}\frac{1}{2Z_n}[\mathrm{Pf}(K_n^{-1,1})_{E^c}
+\mathrm{Pf}(K_n^{1,-1})_{E^c}+\mathrm{Pf}(K_n^{-1,-1})_{E^c}]
=\mathrm{Pf}(K_{\infty}^{-1})_{E}
\end{align*}
and
\begin{align*}
\lim_{n\rightarrow\infty}\frac{\mathrm{Pf}(K_n^{1,1})_{E^c}}{2Z_n}=0
\end{align*}
If we take $E=V(d_{i,j},...,d_{p,j})$, our theorem follows. $\Box$

\section{Simulation of 1-2 Model}

Assume we have 1-2 model with signature
$r_s=\left(\begin{array}{cccccccc}0&c&b&a&a&b&c&0\end{array}\right)^t$
at all vertices. In other words, at each vertex, either one or two
edges are allowed to be present in a local configuration. By Section
3.2, the signature is orthogonally realizable with base change
matrix
$T=\left(\begin{array}{cc}\cos\frac{3\pi}{4}&\sin\frac{3\pi}{4}
\\-\sin\frac{3\pi}{4}&\cos\frac{3\pi}{4}\end{array}\right)$ on all edges.

We define a discrete-time, time-homogeneous Markov chain
$\mathfrak{M}_t$ with state space the set of all configurations of
the 1-2 model. For an $n\times n$ honeycomb lattice embedded into a
torus, the state space is finite, and let us denote it by
$\{i_k\}_{k=1}^{K}$. Let $\Gamma(i_k)$ be the set of configurations
that can be obtained from configuration $i_k$ by adding or deleting
a single edge. Assume $p,q,r$ are edge variables, namely
$p,q,r\in\{a,b,c\}$, and $\{p,q,r\}=\{a,b,c\}$. Define
$\Gamma_{p,+q}(i_k)$ be the set of configurations of 1-2 model which
can be obtained from $i_k$ by adding a single $q$-edge $uv$; before
adding $uv$, only a $p$-edge is present at both $u$ and $v$. Define
$\Gamma_{p,-q}(i_k)$ be the set of configurations of 1-2 model which
can be obtained from $i_k$ by deleting a single $q$-edge $uv$; after
deleting $uv$, only a $p$-edge is present at both $u$ and $v$.
Define
$\Gamma_{0}(i_k)=\Gamma(i_k)\setminus\{\cup_{p,q\in{a,b,c},p\neq
q}\Gamma_{p,+q}(i_k)\cup_{p,q\in{a,b,c},p\neq
q}\Gamma_{p,-q}(i_k)\}$. Define entries of transition matrix for
$\mathfrak{M}_t$ as follows:
\begin{align*}
P(i_l|i_k)=\left\{\begin{array}{cc}\frac{1}{3n^2}&\mathrm{if}\
i_l\in\Gamma_{0}(i_k)\\ \frac{1}{3n^2}&\mathrm{if}\
i_l\in\Gamma_{p,+q}(i_k)\ \mathrm{and}\ r\geq p\\
\frac{1}{3n^2}\frac{r^2}{p^2}&\mathrm{if}\ i_l\in\Gamma_{p,+q}(i_k)\
\mathrm{and}\ r<p\\ \frac{1}{3n^2}&\mathrm{if}\
i_l\in\Gamma_{p,-q}(i_k)\ \mathrm{and}\ r\leq p\\
\frac{1}{3n^2}\frac{p^2}{r^2}&\mathrm{if}\ i_l\in\Gamma_{p,-q}(i_k)\
\mathrm{and}\ r>
p\\1-\sum_{_{i_j\in\Gamma(i_k)}}P(i_j|i_k)&\mathrm{if}\
i_l=i_k\\0&\mathrm{else}\end{array}\right.
\end{align*}
Obviously, $\mathfrak{M}_t$ is aperiodic. For more information on
Markov chain, see \cite{lawler,lpw}. Moreover, we have

\begin{proposition}
$\mathfrak{M}_t$ is irreducible.
\end{proposition}
\begin{proof}
By definition, we only need to prove that any two configurations
communicate to each other. We claim that any two dimer configuration
can obtained from each other by finite steps. In fact, the symmetric
difference of any two dimer configurations is a union of finitely
many loops. Obviously one dimer configuration can be obtained from
any other dimer configuration by first adding finitely many edges to
achieve their union, then deleting alternating edges of loops.
Notice that we have a 1-2 configuration at each step. Hence we only
need to prove that any configuration of 1-2 model can reach a dimer
by finite steps, each of which is adding or deleting one single
edge.
\par Let us start with an arbitrary configuration of 1-2 model. There
are 3 types of connected local configurations: loops with even
number of edges; zigzag paths with odd number of edges; zigzag paths
with even number of edges. For the first and second types, we can
always achieve dimers by deleting alternating edges. Hence we can
assume that all the zigzag paths with even number of edges are of
length 2, and all the other connected local configurations are
dimers. There are two types of length-2 paths. One has vertices
black-white-black (BWB), and the other has vertices
white-black-white (WBW). For each fixed configuration, the number of
BWB paths is the same as the number of WBW paths; otherwise the
complement graph cannot be covered by dimers. Consider an arbitrary
WBW path in a fixed configuration, as illustrated in Figure 12.

\begin{figure}[htbp]
\centering
\scalebox{0.8}[0.8]{\includegraphics*[220,680][388,820]{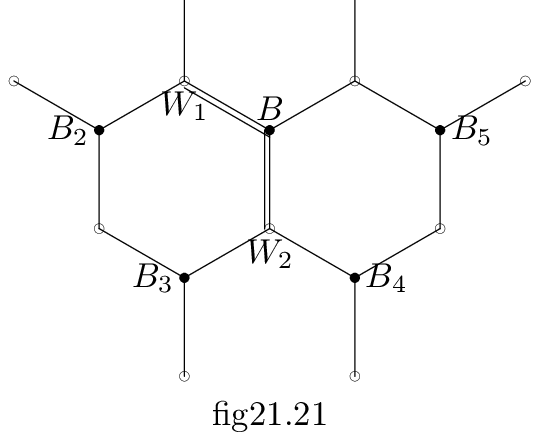}}
\caption{WBW Path}
\end{figure}

One can easily check that no matter what the configuration is, it
communicates with a configuration satisfying one of the following
two conditions: 1. the number of length-2 edges is decreased by 2;
2. it can be moved to a WBW configurations containing $B_i$, for all
$1\leq i\leq 6$. Hence if the number of length-2 paths does not
decrease, one can transverse all black vertices without meeting a
BWB path, because our graph is finite. However, this is impossible
because a BWB path always exists as long as a WBW path exists.
\end{proof}

For an irreducible, aperiodic Markov chain $\mathfrak{M}_t$ with
transition matrix $P$ and stationary distribution $\pi$, let $x_0$
be an arbitrary initial distribution. Then
\begin{align*}
\lim_{n\rightarrow\infty}x_0P^n=\pi.
\end{align*}
Therefore, in order to sample a configuration, we can approximately
sample according to the distribution $x_0P^N$, with large $N$. To
that end, first we choose a fixed dimer configuration with
probability 1 as the initial distribution $x_0$, then we randomly
change the configuration by adding or deleting a single edge
according to the conditional probability specified by the transition
matrix $P$. If neither adding nor deleting $u_1u_2$ ends up with a
satisfying configuration, we just keep the previous configuration;
else we get a new configuration by adding or deleting $u_1u_2$. Then
we repeat the process for $N$ steps. This way we get a sample for
distribution $x_0P^N$, if $N$ is sufficiently large, this is
approximately a sample for distribution $\pi$, which is exactly the
distribution given by 1-2 model.

\begin{example} (Uniform 1-2 Model) Consider the 1-2 model with
$a=b=c=1$. After the holographic reduction, the signature of each
matchgate is
\[\left(\begin{array}{cccccccc}0&\frac{\sqrt 2}{2}&\frac{\sqrt 2}{2}&0&\frac{\sqrt{2}}{2}&0&0&-\frac{3\sqrt2}{2}\end{array}\right)'\] It
is gauge equivalent to a positive-weight dimer model on Fisher
graph, whose spectral curve does not intersect $\mathbb{T}^2$. We
are interested in the probability that a $\{001\}$ dimer occurs,
that is, at a pair of adjacent vertices $v_1,v_2$ connected by an
$a$-edge, only the configuration $\{001\}$ is allowed. By the
technique described in section 5 and section 6, the partition
function of the configurations in which a dimer is present at
$v_1v_2$ is equal to the sum of two partition functions, one
corresponds to both $v_1$ and $v_2$ are replaced by an odd matchgate
with signatures
\[\left(\begin{array}{cccccccc}0&-\frac{\sqrt{2}}{4}&\frac{\sqrt{2}}{4}&0&\frac{\sqrt{2}}{4}&0&0&-\frac{\sqrt{2}}{4}\end{array}\right);\]
the other corresponds to both $v_1$ and $v_2$ are replaced by an
even matchgate with signatures
\[\left(\begin{array}{cccccccc}\frac{\sqrt{2}}{4}&0&0&-\frac{\sqrt{2}}{4}&0&-\frac{\sqrt{2}}{4}&\frac{\sqrt{2}}{4}&0\end{array}\right)\]
Both of them have the same partition function as a graph with
positive weights. Moreover, if we give a base change matrix
$\left(\begin{array}{cc}0&1\\1&0\end{array}\right)$ on $v_1v_2$
edge, and apply holographic reduction, we see that the two graphs
with positive weights are holographic equivalent. Hence it suffices
to consider the one with a pair of odd matchgates with modified
weights. Then
\[Pr(\mathrm{a\ dimer\ is\ present\ at}\ v_1v_2)=\frac{Z_{001}}{18Z}\]
where $Z_{001}$ is the partition function of dimer configurations
with weight 1 on triangles corresponding to $v_1$ and $v_2$ and
weight $\frac{1}{3}$ on all the other triangles, and $Z$ is the
partition function with weight $\frac{1}{3}$ on all the triangles.
Moreover
\begin{align*}
Z_{001}&=&Z^{001001}+Z^{001111}+Z^{010010}+Z^{010100}+Z^{100100}+Z^{100010}+Z^{111001}+Z^{111111}\\
&=&Z^{001001}+2Z^{001111}+2Z^{010010}+2Z^{010100}+Z^{111111}
\end{align*}
where $Z^{ijk,\tilde{i}\tilde{j}k}$ is the partition function of
dimer configurations with fixed configuration
$ijk,\tilde{i}\tilde{j}k$ on triangles corresponding to $v_1$ and
$v_2$. The second equality follows from symmetry. Meanwhile
\[\frac{1}{9}Z^{001001}+\frac{2}{3}Z^{001111}+\frac{2}{9}Z^{010010}+\frac{2}{9}Z^{010100}+Z^{111111}=Z,\]
then
\begin{align*}
Pr(\mathrm{a\ dimer\ is\ present\ at}\
v_1v_2)&=&\frac{1}{18}\left(1+\frac{8Z^{001001}}{9Z}+\frac{4Z^{001111}}{3Z}+\frac{16Z^{010010}}{9Z}+\frac{16Z^{010100}}{9Z}\right)\\
&=&\frac{1}{18}\left(1+\frac{4}{3}|(K^{-1}_{\infty})_{23}|+\frac{4}{9}|\mathrm{Pf}(K^{-1}_{\infty})_{2356}|+\frac{16}{3}|(K^{-1}_{\infty})_{13}|\right)
\end{align*}
where
\begin{align*}
(K^{-1}_{\infty})_{23}&=&\frac{3}{16\pi^2}\iint_{\mathbb{T}^2}\frac{w(zw+1+z-9w)}{2(z^2+w^2)+2(z+w)+2zw(z+w)-21zw}\frac{dz}{iz}\frac{dw}{iw}=(K^{-1}_{\infty})_{56}\\
(K^{-1}_{\infty})_{13}&=&-\frac{3}{16\pi^2}\iint_{\mathbb{T}^2}\frac{w(w+z+z^2-9zw)}{2(z^2+w^2)+2(z+w)+2zw(z+w)-21zw}\frac{dz}{iz}\frac{dw}{iw}\\
(K^{-1}_{\infty})_{25}&=&\frac{1}{16\pi^2}\iint_{\mathbb{T}^2}\frac{8z+8zw-82w+9w^2+9}{2(z^2+w^2)+2(z+w)+2zw(z+w)-21zw}\frac{dz}{iz}\frac{dw}{iw}=(K^{-1}_{\infty})_{36}\\
(K^{-1}_{\infty})_{26}&=&\frac{1}{16\pi^2}\iint_{\mathbb{T}^2}\frac{-z-zw-w+9}{2(z^2+w^2)+2(z+w)+2zw(z+w)-21zw}\frac{dz}{iz}\frac{dw}{iw}=(K^{-1}_{\infty})_{35}
\end{align*}
The entries of the inverse matrix can be expressed as elliptic
functions. The probability that a $\{001\}$ dimer occurs is
approximately $6\%$. A sample of the uniform 1-2 model is
illustrated in Figure 16.

\begin{figure}[htbp]
\centering
\rotatebox{90}{\scalebox{1.2}[1.2]{\includegraphics*[38,253][593,553]{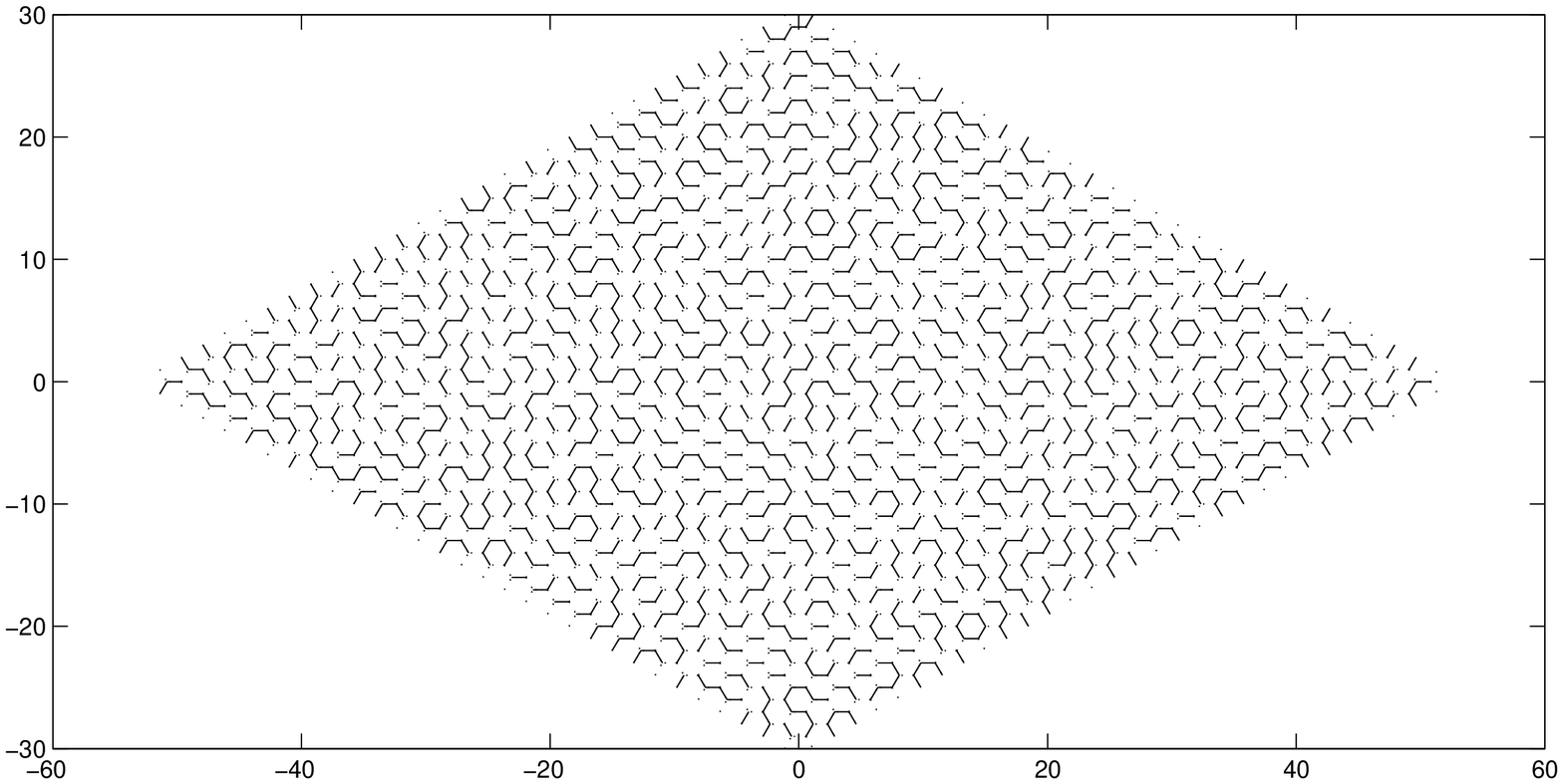}}}
\caption{Sample of Uniform 1-2 Model}
\end{figure}
\end{example}

\begin{example} (Critical 1-2 Model) Consider the 1-2 model with
$a=4,b=c=1$. After the holographic reduction, it has the same
partition function as a positive-weight dimer model on Fisher graph,
whose spectral curve has a single real node on $\mathbb{T}^2$. The
probability of the configuration $\{011\}$
\[Pr(\{011\})=\frac{Z_{011}}{3Z}=\frac{1}{3}(|(K^{-1}_{\infty})_{12}|+|(K^{-1}_{\infty})_{23}|+|(K^{-1}_{\infty})_{13}|)+|\mathrm{Pf}(K^{-1}_{\infty})_{123456}|)\]
$Z_{011}$ is the partition function on a Fisher graph with weights
$1,1,1$ on one triangle, and weights
$\frac{2}{3},\frac{2}{3},\frac{1}{3}$ on all the other triangles.
$Z$ is the partition function with weights
$\frac{2}{3},\frac{2}{3},\frac{1}{3}$ on all the triangles. By
symmetry $|(K^{-1}_{\infty})_{12}|=|(K^{-1}_{\infty})_{23}|$. Then
\begin{eqnarray*}
Pr(\{011\})&=&\frac{1}{3}+\frac{2}{9}|(K^{-1}_{\infty})_{12}|+\frac{2}{9}|(K^{-1}_{\infty})_{23}|\\
&=&\frac{1}{3}+\frac{1}{6\pi^2}\left|\iint_{\mathbb{T}^2}\frac{w(4wz+4+z-9w)}{-7(w^2+z^2)+32wz(w+z)+32(w+z)-114wz}\frac{dz}{iz}\frac{dw}{iw}\right|\\
&&+\frac{1}{3\pi^2}\left|\iint_{\mathbb{T}^2}\frac{z(9wz-4z-w^2-4w)}{-7(w^2+z^2)+32wz(w+z)+32(w+z)-114wz}\frac{dz}{iz}\frac{dw}{iw}\right|\\
&=&\frac{1}{3}+\frac{1}{12\pi}\left|\int_{|z|=1}\frac{8}{32z-7}-\frac{18}{z(32z-7)}+\frac{16z^2-85z+29}{z(32z-7)\sqrt{4z^2-17z+4}}dz\right|\\
&&+\frac{1}{3\pi}\left|\int_{|w|=1}\frac{9}{(32w-7)}-\frac{4}{w(32w-7)}+\frac{22w^2-35w+8}{w(32w-7)\sqrt{4w^2-17w+4}}dw\right|\\
&=&\frac{23}{48}-\frac{25}{112\pi}\arctan{\frac{4}{3}}-\frac{65}{336\pi}\arctan{\frac{44}{117}}\approx39\%.
\end{eqnarray*}
For $\sqrt{\cdot}$ we choose a branch with positive real part. This
integral can be evaluated explicitly because the graph has critical
edge weights. A sample of the critical 1-2 model is illustrated in
Figure 17.

\begin{figure}[htbp]
\centering
\scalebox{0.7}[0.67]{\includegraphics*[98,211][690,578]{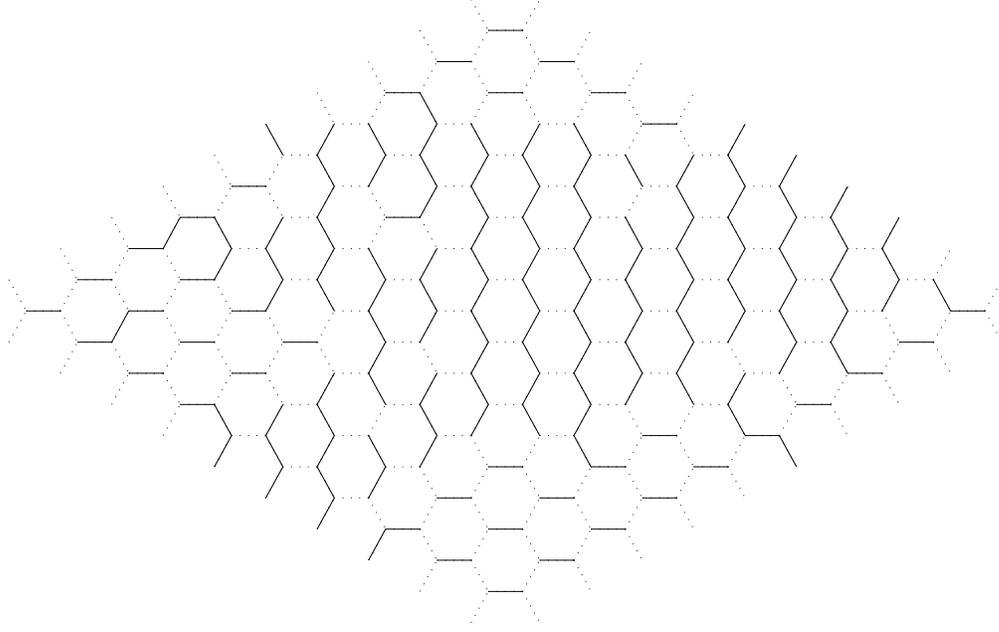}}
\caption{Sample of Critical 1-2 Model}
\end{figure}

\end{example}

\textbf{Question}: How large should $N$ be as a function of the size
of the graph?

\appendix
\section{Realizability Conditions}

In this section we give explicit realizability conditions for
periodic honeycomb lattice embedded on a $n\times n$ domain. Let us
consider an arbitrary vertex $x$ with adjacent vertices $y,z,w$.
Assume relations have following signatures:
\begin{align*}
r_x=\left(\begin{array}{ccc}x_1&...&x_8\end{array}\right)^t\ \ \
r_y=\left(\begin{array}{ccc}y_1&...&y_8\end{array}\right)^t\\
r_z=\left(\begin{array}{ccc}z_1&...&z_8\end{array}\right)^t\ \ \
r_w=\left(\begin{array}{ccc}w_1&...&w_8\end{array}\right)^t
\end{align*}

According to the realizability equation (1), after a lengthy
computation, we have
\\
\begin{theorem}
A periodic network of relation on $n\times n$ honeycomb lattice is
relizable if and only if at each vertex $x$ with adjacent vertices
$y,z,w$ connected by $a,b,c$ edge, respectively, the signatures
satisfy the following equation:
\begin{align*}
\sum_{i,j,k\in\{0,1,2\}}X_{ijk}Y_iZ_jW_k=0
\end{align*}
where
\begin{align*}
\begin{array}{ccc}Y_{0}=y_1y_4-y_2y_3&Y_1=y_1y_8+y_4y_5-y_2y_7-y_3y_6&Y_2=y_5y_8-y_6y_7\\Z_0=z_1z_6-z_2z_5&Z_1=z_1z_8+z_3z_6-z_4z_5-z_2z_7&Z_2=z_3z_8-z_4z_7\\
W_0=w_1w_7-w_3w_5&W_1=w_1w_8+w_2w_7-w_5w_4-w_3w_6&W_2=w_2w_8-w_6w_4
\end{array}
\end{align*}
and
\begin{align*}
X_{000}&=&x_1^2(x_3x_6+x_2x_7+x_4x_5-x_1x_8)-2x_1x_2x_3x_5\\
X_{001}&=&x_1^2(x_6x_4-x_2x_8)+x_2^2(x_1x_7-x_3x_5)\\
X_{002}&=&x_2^2(x_2x_7-x_1x_8-x_3x_6-x_4x_5)+2x_1x_2x_4x_6\\
X_{010}&=&x_3^2(x_1x_6-x_2x_5)+x_1^2(x_7x_4-x_3x_8)\\
X_{011}&=&x_1x_2(x_4x_7-x_3x_8)+x_3x_4(x_1x_6-x_2x_5)\\
X_{012}&=&x_2^2(x_4x_7-x_3x_8)+x_4^2(x_1x_6-x_2x_5)\\
X_{020}&=&x_3^2(x_3x_6-x_2x_7-x_4x_5-x_1x_8)+2x_1x_4x_3x_7\\
X_{021}&=&x_4^2(x_1x_7-x_3x_5)+x_3^2(x_6x_4-x_2x_8)\\
X_{022}&=&x_4^2(x_2x_7+x_1x_8+x_3x_6-x_4x_5)-2x_2x_3x_4x_8\\
\end{align*}
\begin{align*}
X_{100}&=&x_5^2(x_1x_4-x_2x_3)+x_1^2(x_6x_7-x_5x_8)\\
X_{101}&=&x_1x_2(x_6x_7-x_5x_8)+x_5x_6(x_1x_4-x_2x_3)\\
X_{102}&=&x_6^2(x_1x_4-x_2x_3)+x_2^2(x_6x_7-x_5x_8)\\
X_{110}&=&x_1x_3(x_7x_6-x_5x_8)+x_5x_7(x_1x_4-x_2x_3)\\
X_{111}&=&x_1x_4x_6x_7-x_2x_3x_5x_8\\
X_{112}&=&x_4x_8(x_1x_6-x_2x_5)+x_2x_6(x_4x_7-x_3x_8)\\
X_{120}&=&x_7^2(x_1x_4-x_2x_3)+x_3^2(x_6x_7-x_5x_8)\\
X_{121}&=&x_3x_7(x_4x_6-x_2x_8)+x_4x_8(x_1x_7-x_3x_5)\\
X_{122}&=&x_8^2(x_1x_4-x_2x_3)+x_4^2(x_6x_7-x_5x_8)\\
\end{align*}
\begin{align*}
X_{200}&=&x_5^2(x_4x_5-x_1x_8-x_3x_6-x_2x_7)+2x_1x_5x_6x_7\\
X_{201}&=&x_5^2(x_6x_4-x_2x_8)+x_6^2(x_1x_7-x_3x_5)\\
X_{202}&=&x_6^2(x_1x_8+x_2x_7+x_4x_5-x_3x_6)-2x_5x_6x_2x_8\\
X_{210}&=&x_5^2(x_4x_7-x_3x_8)+x_7^2(x_1x_6-x_2x_5)\\
X_{211}&=&x_6x_8(x_1x_7-x_3x_5)+x_5x_7(x_4x_6-x_2x_8)\\
X_{212}&=&x_8^2(x_1x_6-x_2x_5)+x_6^2(x_4x_7-x_3x_8)\\
X_{220}&=&x_7^2(x_1x_8+x_3x_6+x_4x_5-x_2x_7)-2x_3x_8x_5x_7\\
X_{221}&=&x_8^2(x_1x_7-x_3x_5)+x_7^2(x_6x_4-x_2x_8)\\
X_{222}&=&x_8^2(x_1x_8-x_3x_6-x_2x_7-x_4x_5)+2x_7x_8x_6x_4
\end{align*}

\end{theorem}

\begin{theorem}
A periodic network of relation with $n\times n$ fundamental domain
is bipartite realizable if and only if it is realizable and at any
vertex $v$, the signature satisfy
\begin{align*}
v_1^2v_8^2+v_2^2v_7^2+v_3^2v_6^2+v_4^2v_5^2-2v_1v_8v_2v_7-2v_1v_8v_3v_6-2v_1v_8v_4v_5\\
-2v_2v_7v_3v_6-2v_2v_7v_4v_5-2v_3v_6v_4v_5+4v_1v_4v_6v_7+4v_2v_3v_5v_8=0
\end{align*}
\end{theorem}
\begin{proof}
Without loss of generality, we assume at each matchgate, the
signature $\{111\}$ is 0, and assume $v$ is a black vertex. The
$\{111\}$ entry of the signature of a black vertex is 0 gives us
\begin{eqnarray}
a_{11}a_{12}a_{13}v_1+a_{11}a_{12}v_2+a_{11}a_{13}v_3+a_{11}v_4+a_{12}a_{13}v_5+a_{12}v_6+a_{13}v_7+v_8=0
\label{brealizability1}
\end{eqnarray}
The parity constraint implies that the $\{110\}$ entry is also 0,
namely
\begin{eqnarray}
a_{11}a_{12}a_{03}v_1+a_{11}a_{12}v_2+a_{11}a_{03}v_3+a_{11}v_4+a_{12}a_{03}v_5+a_{12}v_6+a_{03}v_7+v_8=0\label{brealizability2}
\end{eqnarray}
 From \eqref{brealizability1}\eqref{brealizability2}, we can solve $a_11$, and solution has following
form
\begin{align*}
a_{11}=\frac{N_1}{D_1}=\frac{N_2}{D_2}
\end{align*}
Then $N_1D_2-N_2D_1=0$. Under the assumption that $a_{03}\neq
a_{13}$, we have
\begin{eqnarray}
(-v_2v_5+v_6v_1)a_{12}^2+(v_8v_1+v_6v_3-v_4v_5-v_2v_7)a_{12}+v_8v_3-v_4v_7=0\label{brealizability3}
\end{eqnarray}
From (12), we have
\begin{eqnarray}
2(v_1v_6-v_2v_5)a_{02}a_{12}+(v_1v_8+v_3v_6-v_4v_5-v_2v_7)(a_{02}+a_{12})+2(v_3v_8-v_4v_7)=0\label{brealizability4}
\end{eqnarray}
2$\times$\eqref{brealizability3}$-$\eqref{brealizability4}, under
the assumption that $a_{02}\neq a_{12}$, we have
\begin{align*}
a_{12}=-\frac{v_1v_8+v_3v_6-v_4v_5-v_2v_7}{2(v_1v_6-v_2v_5)}
\end{align*}
which implies that equation \eqref{brealizability3} has double real
roots, and its discriminant is 0. That is exactly the statement of
the theorem.
\end{proof}

\textbf{Proof of Proposition 3.9} Since holographic reduction is an
invertible process, two matchgrids $M$ and $\hat{M}$ are
holographically equivalent if and only if there exists a basis for
each edge, such that one can be transformed to the other using
holographic reduction. Obviously holographic equivalent matchgrids
have the same dimer partition function. Assume weights $m_b,m_w$ of
$M$ and $\hat{m_b},\hat{m_w}$ of $\hat{M}$ are as follows:
\begin{align*}
m_w^{ij}&=\left(\begin{array}{cccccccc}0&c_1^{ij}&b_1^{ij}&0&a_1^{ij}&0&0&d_1^{ij}\end{array}\right)^t\\
m_b^{ij}&=\left(\begin{array}{cccccccc}0&c_2^{ij}&b_2^{ij}&0&a_2^{ij}&0&0&d_2^{ij}\end{array}\right)^t\\
\hat{m}_w^{ij}&=\left(\begin{array}{cccccccc}0&\hat{c}_1^{ij}&\hat{b}_1^{ij}&0&\hat{a}_1^{ij}&0&0&\hat{d}_1^{ij}\end{array}\right)^t\\
\hat{m}_b^{ij}&=\left(\begin{array}{cccccccc}0&\hat{c}_2^{ij}&\hat{b}_2^{ij}&0&\hat{a}_2^{ij}&0&0&\hat{d}_2^{ij}\end{array}\right)^t
\end{align*}
Then by equations (\ref{s1})-(\ref{s6}), and the uniqueness of the
basis on each edge, we have
\begin{eqnarray}
\frac{a_1^{ij}b_1^{ij}}{d_1^{ij}c_1^{ij}}&=&\frac{d_2^{i,j-1}c_2^{i,j-1}}{b_2^{i,j-1}a_2^{i,j-1}}=a_{03}^{ij}a_{13}^{ij}=b_{03}^{i,j-1}b_{13}^{i,j-1}\label{simplifiedequation1}\\
\frac{a_1^{ij}c_1^{ij}}{b_1^{ij}d_1^{ij}}&=&\frac{d_2^{i-1,j}b_2^{i-1,j}}{a_2^{i-1,j}c_5^{i-1,j}}=a_{02}^{ij}a_{12}^{ij}=b_{02}^{i-1,j}b_{12}^{i-1,j}\\
\frac{b_1^{ij}c_1^{ij}}{a_1^{ij}d_1^{ij}}&=&\frac{d_2^{ij}a_2^{ij}}{b_2^{ij}c_2^{ij}}=a_{01}^{ij}a_{11}^{ij}=b_{01}^{ij}b_{11}^{ij}\label{simplifiedequation2}
\end{eqnarray}
Plugging in (\ref{simplifiedequation1})-(\ref{simplifiedequation2})
to (\ref{blackrealizability})-(\ref{whiterealizability}), we have
\begin{align*}
\left\{\begin{array}{c}\hat{c}^{ij}_2={c^{ij}_2}{n_{01}^{i,j,0}n_{02}^{i,j,0}p_{13}^{i,j,0}}\cdot
C_1^{ij}\\
\hat{b}^{ij}_2={b^{ij}_2}{n_{01}^{i,j,0}p_{12}^{i,j,0}n_{03}^{i,j,0}}\cdot C_1^{ij}\\
\hat{a}^{ij}_2={a^{ij}_2}{p_{11}^{i,j,0}n_{02}^{i,j,0}n_{03}^{i,j,0}}\cdot
C_1^{ij}\\
\hat{d}^{ij}_2={d^{ij}_2}{p_{11}^{i,j,0}p_{12}^{i,j,0}p_{13}^{i,j,0}}\cdot
C_1^{ij}\end{array}\right.
\end{align*}
\begin{align*}
\left\{\begin{array}{c}\hat{c}^{ij}_1=\frac{c^{ij}_1}{n_{01}^{i,j,1}n_{02}^{i,j,1}p_{13}^{i,j,1}}\cdot C_2^{ij}\\
\hat{b}^{ij}_1=\frac{b^{ij}_1}{n_{01}^{i,j,1}p_{12}^{i,j,1}n_{03}^{i,j,1}}\cdot C_2^{ij}\\
\hat{a}^{ij}_1=\frac{a^{ij}_1}{p_{11}^{i,j,1}n_{02}^{i,j,1}n_{03}^{i,j,1}}\cdot C_2^{ij}\\
\hat{d}^{ij}_1=\frac{d^{ij}_1}{p_{11}^{i,j,1}p_{12}^{i,j,1}p_{13}^{i,j,1}}\cdot
C_2^{ij}\end{array}\right.
\end{align*}
To prove that probability measures of $M$ and $\hat{M}$ are
identical, we only need to prove that for any dimer configuration,
the products of weights differ by the same constant factor. Each
dimer configuration corresponds to a binary sequence of length $N$,
where $N=3n^2$ is the number of connecting edges. Choose an
arbitrary edge with basis
$\left(\begin{array}{cc}n_0&n_1\\p_0&p_1\end{array}\right)$. If the
edge is occupied, then from $M$ to $\hat{M}$, adjacent generator
weight is divided by $p_1$, while adjacent recognizer weight is
multiplied by $p_1$. If the edge is unoccupied, then from $M$ to
$\hat{M}$, adjacent generator weight is divided by $n_0$, while the
adjacent recognizer weight is multiplied by $n_0$. Therefore, the
total effect is for any particular dimer configuration $\varpi$
\begin{align*}
\frac{\mathrm{Partition}(\varpi\ on\
\hat{M})}{\mathrm{Partition}(\varpi\ on\
M)}=\prod_{i,j}C_1^{ij}C_{2}^{ij}
\end{align*}
which is a constant independent of configuration $\varpi$. $\Box$


\begin{thebibliography}{99}
\bibitem{bt}C.~Boutillier and B.~Tili\`{e}re, the Critical Z-invariant
Ising Model via Dimers: the Periodic Case, arxiv:math/0812.0348
\bibitem{cai}J.~Cai, Holographic Algorithms, Current Developments in
Mathematics, 2005(2007),111-150
\bibitem{ckp}H.~Cohn, R.~Kenyon, J.~Propp, a Variational Principle for
Domino Tilings, J. Amer. Math. Soc., 14(2001), no.2, 297-346
\bibitem{ka1}P.~W.~Kasteleyn, the Statistics of Dimers on a Lattice,
Physica, 27(1961), 1209-1225
\bibitem{ka2}P.~W.~Kasteleyn, Graph Theory and Crystal Physics, in
Graph Theory and Theoretical Physics, Academic Press, London, 1967
\bibitem{ke1}R.~Kenyon, an Introduction to the Dimer Model,
Math.CO/0310326
\bibitem{ke2}R.~Kenyon, Local Statistics on Lattice Dimers, Ann.
Inst. H. Poicar$\acute{e}$. Probabilit$\acute{e}$s, 33(1997),
591-618
\bibitem{kos}R.~Kenyon, A.~Okounkov, S.Sheffield, Dimers and Amoebae,
Ann. Math. 163(2006), no.3, 1019-1056
\bibitem{ko}R.~Kenyon, A.~Okounkov, Planar Dimers and Harnack Curve,
Duke. Math. J. 131(2006), no.3, 499-524
\bibitem{kir}A.~Kirillov,~Jr.,Introduction of Lie Groups and Lie
Algebra, Cambridge Studies in Advanced Mathematics, no.113(2008)
\bibitem{lawler}G.~F.~Lawler, Introduction to Stochastic Processes,
2nd Edition, Chapman Hall/CRC (2006)
\bibitem{lpw}D.~Levin, Y.~Peres, E.~Wilmer, Markov Chains and
Mixing Times, American Mathematical Society (2008)
\bibitem{li}Z.~Li, Spectral Curve of a Periodic Fisher Graph, in
preparation
\bibitem{sb}M.~Schwartz and J.~Bruck, Constrained Codes as Network of
Relations, Information Theory, IEEE Transactions, 54(2008), Issue 5,
2179-2195
\bibitem{gt}G.~Tesler, Matchings in Graphs on Non-orientable
Surfaces, J.~Combin.~Theory Ser.~B 78(2000), no.~2,198-231
\bibitem{val}L.~G.~Valiant, Holographic Algorithms(Extended Abstract),
in Proc. 45th IEEE Symposium on Foundations of Computer
Science(2004), 306-315
\end{thebibliography}
\end{document}